\documentclass[12pt]{article} 
\usepackage[sectionbib]{natbib}
\usepackage{array,epsfig, fancyhdr,rotating}
\usepackage[]{hyperref}  

\usepackage{amsmath}
\usepackage{amssymb}
\usepackage{amsfonts}
\usepackage{multirow}
\usepackage{amsthm}
\usepackage{microtype}
\usepackage{graphicx}
\usepackage{subfigure}
\usepackage{booktabs} 
\usepackage{enumerate}
\usepackage{mathrsfs}
\usepackage{bbm}
\usepackage{hyperref}
\usepackage{enumitem}
\usepackage{mathtools}
\usepackage[margin=1in]{geometry}

\newcommand{\andd}{\qquad\text{and}\qquad}

\newcommand{\ind}[1]{\mathbbm{1}\left\{#1\right\}}

\newcommand{\rdd}{\mathbb{R}^{d}}
\newcommand{\re}{\mathbb{R}}

\newcommand{\Prr}[1]{\Pr\left(#1\right)}

\newcommand{\sam}[2]{\mathbb{#1}_{#2}}

\newcommand{\eqd}{\stackrel{\text{d}}{=}}

\newcommand{\D}{{\rm D}}

\newcommand{\E}{\mathbb{E}}

\DeclarePairedDelimiter\floor{\lfloor}{\rfloor}

\DeclareMathOperator{\var}{Var}

\newcommand{\GS}{\mathrm{GS}}

\newcommand{\cM}{\mathcal{M}}
\newcommand{\cD}{\mathcal{D}}

\newcommand{\cA}{\mathcal{A}}
\newcommand{\cN}{\mathcal{N}}

\newcommand{\bfX}{\mathbf{X}}

\newcommand{\bfY}{\mathbf{Y}}
\newcommand{\bfx}{\mathbf{x}}
\newcommand{\bfz}{\mathbf{z}}
\newcommand{\bfy}{\mathbf{y}}

\newcommand{\bbN}{\mathbb{N}}

\newtheorem{theorem}{Theorem}
\newtheorem{lemma}{Lemma}
\newtheorem{corollary}{Corollary}
\newtheorem{proposition}{Proposition}
\theoremstyle{definition}
\newtheorem{definition}{Definition}
\newtheorem{condition}{Condition}

\newtheorem{remark}{Remark}
\title{Improved subsample-and-aggregate via the private modified winsorized mean}
\author{Kelly Ramsay and Dylan Spicker}
\begin{document}


\maketitle	
	
\begin{abstract}
We develop a univariate, differentially private mean estimator, called the \textit{private modified winsorized mean}, designed to be used as the aggregator in subsample-and-aggregate. 
We demonstrate, via real data analysis, that common differentially private multivariate mean estimators may not perform well as the aggregator, even in large datasets, motivating our developments.
We show that the modified winsorized mean is minimax optimal for several, large classes of distributions, even under adversarial contamination. 
We also demonstrate that, empirically, the private modified winsorized mean performs well compared to other private mean estimates. 
We consider the modified winsorized mean as the aggregator in subsample-and-aggregate, deriving a finite sample deviations bound for a subsample-and-aggregate estimate generated with the new aggregator. 
This result yields two important insights: (i) the optimal choice of subsamples depends on the bias of the estimator computed on the subsamples, and (ii) the rate of convergence of the subsample-and-aggregate estimator depends on the robustness of the estimator computed on the subsamples. 
\vspace{9pt}
\par
\end{abstract}

\section{Introduction}
Subsample-and-aggregate is a powerful algorithm that can produce a differentially privately version of any statistic \citep{Nissim2007}. 
For instance, through subsample-and-aggregate, complex machine learning models can be made private easily, with minimal additional programming. 
The algorithm is a two-step procedure, relying on a ``subsample'' step and an ``aggregate'' step. 
The idea, in brief, is to first split the data into disjoint subsets. 
Then, on each of these subsets, a nonprivate version of the desired statistic is computed.
Next, the set of nonprivate statistics are combined, using a differentially private aggregator, producing a private estimate of the parameter of interest. 
Generally, if the mean of the subsamples is desired, the private aggregator is a basic differentially private clipped mean \citep{Nissim2007} or a widened winsorized mean \citep{smith2011privacy}. 
However, to achieve pure differential privacy, both of these aggregators require tight input bounds on the nonprivate statistics. 

It is natural to consider more recent, differentially private mean estimators as the private aggregator. These have a weaker dependence, if any, on input bounds. 
However, currently, replacing the aggregator with any of the existing differentially private mean estimators comes with at least one of the following drawbacks: (i) only a weaker notion of differential privacy is satisfied, (ii) many subsamples are required for a stable mean estimate, resulting in large sample size requirements, (iii) the new algorithm has strong dependence on tuning parameters, (iv) strong assumptions on, or knowledge of, the population distribution is required, or (v) the estimator is not practically computable. 

Here, we develop a univariate, differentially private mean estimator, called the \textit{private modified winsorized mean}, or, PMW mean. 
Our estimator is designed to be used as the aggregator in subsample-and-aggregate and avoids these drawbacks. 
The PMW mean can satisfy either pure or zero-concentrated differential privacy, performs well in small samples or under adversarial contamination, has weak dependence on input bounds and other tuning parameters, and can be used to estimate the mean of a general distribution, with only two moments. 
The proposed estimator combines the winsorized mean of \citet{Lugosi2021} and the private quantile estimation procedure of \citet{Durfee2024}. 

\subsection{Related work}
Differentially private mean estimation has been well-studied. 
In the univariate setting, clipped and trimmed means have been the main approach \citep{Barber2014, Karwa2017, Bun2019, Kamath2020, tsfadia2022, Durfee2024}. 
The clipped mean procedure operates by selecting an interval $[a,b]$, either empirically or \textit{a priori}, and projecting the dataset onto the interval.
An additive noise mechanism is then used to ensure privacy. 
The trimmed mean functions similarly, removing extreme observations instead of projecting them.
Existing estimators using these approaches have at least one of the previously mentioned drawbacks.
In particular, \citet{tsfadia2022} consider approximately differentially private mean estimation. 
For stronger forms of privacy, many mean estimation algorithms require certain components of the population distribution as algorithm inputs, such as the number of moments \citep{Barber2014, Kamath2020} or a lower bound on the variance\citep{Bun2019}. 
Conversely, algorithms aimed that they can be implemented without knowledge of the population distribution lack minimax optimality or similar statistical guarantees \citep{tsfadia2022, Durfee2024}. 
The most similar work is that of \citet{Bun2019}, who develop a zero-concentrated differentially private trimmed mean estimate, combined with smooth sensitivity. 
Their estimator avoids many of the drawbacks listed, however, a lower bound on the variance is required, and it is unclear how to choose the smoothing parameter. 
In addition, their approach removes data from the estimator, which can be problematic at low sample sizes (see Section~\ref{sec::simulation2}).

Multivariate estimators can be applied to univariate estimation problems. 
In the multivariate setting, approximate differentially private mean estimation has been well-studied \citep{2021Liub,Brown2021,Esfandiari2022,Duchi2023,Brown2023,Dagan2024}. 
Under stronger privacy guarantees, many authors have considered mean estimators for a sample of subgaussian data \citep{Kamath2018, Bun2019b,Cai2019,Biswas2020,2021Liu}. 
Several authors have relaxed this assumption using moment conditions on the population distribution, through contamination models, or both \citep{Hopkins2021,Yu2024}. 
Unfortunately, some of these algorithms are not currently practically computable \citep{Hopkins2021,Brown2021,2021Liu}. 
For the computable estimators, at small sample sizes ($n<1000$) we have observed empirically that they do not perform as well as univariate estimators applied coordinate-wise, and may be unstable or highly sensitive to tuning parameters (see Section~\ref{sec::data-analysis}). 
To use these estimators in subsample-and-aggregate, we would then require thousands of subsamples, and, thus, a very large sample size. 

Beyond mean estimation, several authors have considered variants of subsample-and-aggregate \citep{smith2011privacy, bassily2018model,jordon2019differentially, amin2022easy}. 
Most of these focus on the ``subsample'' step of the algorithm \citep{papernot2016semi, bassily2018model,jordon2019differentially}. 
\citet{amin2022easy} and \citet{smith2011privacy} consider alternatives for the private aggregator. 
\citet{smith2011privacy} show that using a private widened winsorized mean as the aggregator results in the same asymptotic distribution as the nonprivate estimator used in the subsample step.
However, this aggregator is sensitive to input bounds, unless we relax the setting to approximate differential privacy.
\citet{amin2022easy} consider an approximately differentially private, multivariate median aggregator in the setting of linear regression. 

\subsection{Contributions}

Our main contributions are: (i) We develop a new, univariate, differentially private mean estimator, designed to be used as the aggregator in subsample-and-aggregate. (ii) We show that this estimator is minimax optimal for several large classes of distributions, even under adversarial contamination (Theorem~\ref{thm::main-result}). As a corollary, we present a finite sample deviation bound for subsample-and-aggregate with the new aggregator (Corollary~\ref{corr:ssa}). (iii) We demonstrate that the estimator performs well relative to competitors in simulation (Section~\ref{sec::simulation2}). (iv) We demonstrate via, data analysis, that common differentially private multivariate mean estimators do not perform well as the aggregator, even with data with 2300 subjects (Section~\ref{sec::data-analysis}). (v) As a by-product of our analysis, we introduce a zero-concentrated differentially private version of the unbounded private quantile algorithm \citep{Durfee2024}, as well as concentration inequalities for these quantile estimates (Lemma~\ref{lem::PQE-zCDP} and Lemma~\ref{lem::pq-bound}--\ref{lem::pq-bound-zCDP}). In addition, we correct an error in the proof of Theorem 1 of \citet{Lugosi2021} (Appendices~\ref{app::lm-error} and \ref{sec::proof-main}). 


\section{Preliminary material}\label{sec::prelim}
\subsection{Problem statement}
First, we formalize the problem.
Suppose that we have a random sample of size $2n$, with at most $2\eta n$ points, ($0 \leq \eta \leq 1/2$), \emph{corrupted} in that they have been changed to arbitrary values, where the corrupted values may depend on the uncorrupted values. 
Denote the corrupted sample $X_1',\ldots,X_{n}' ,Y_1',\ldots,Y_n'$. 
Assume that the uncorrupted points, denoted $X_1,\ldots,X_{n},Y_1,\ldots,Y_n$, are independent and identically distributed according to a univariate distribution $F$ with mean $\mu\in\re$ and variance  $\sigma^2<\infty$. 
The aim is to estimate $\mu$ while maintaining differential privacy (defined in Section~\ref{sec::dp}). 
This setup is summarized as Condition~\ref{cond::problem}.

\begin{condition}\label{cond::problem}
Take an independent, identically distributed sample  $X_1,\ldots,X_{n},Y_1,\ldots,Y_n$, drawn from a univariate distribution $F$ with mean $\mu\in\re$ and variance $\sigma^2<\infty$, to be the uncorrupted points.
The points $X_1',\ldots,X_{n}',Y_1',\ldots,Y_n'$ then differ arbitrarily from the uncorrupted points at at most $2\eta n$ points, $0 \leq \eta \leq 1/2$, where the corrupted values may depend on the uncorrupted values.
\end{condition}

\subsection{Differential privacy}\label{sec::dp}
Next, we introduce the concepts from differential privacy used in this work \citep{Dwork2006, Dwork2014}. 
Define a dataset of size $n\in \bbN$ to be a set of $n$ real numbers, and let $\cD_n$ be the set of datasets of size $n$. 
We say that a dataset $\bfx_n\in \cD_n$, is adjacent to another dataset $\bfy_n\in \cD_n$ if $\bfx_n$ and $\bfy_n$ differ by exactly one point. 
Let $\cA_n$ be the set of pairs of adjacent datasets of size $n$. 
Let $G_{\bfx_n}$ be a probability distribution over $\rdd$ with $d\in\bbN$ that depends on $\bfx_n$. 
That is, $G_{.}\colon \cD_n\to \cM_1(\rdd),$ where $\cM_1(\rdd)$ is the set of (Borel) probability distributions on $\rdd$. 
If $G_{\bfx_n}$ is absolutely continuous, let $g_{\bfx_n}$ be the associated density. 
For $\alpha>1$, denote the $\alpha$-R\'enyi Divergence between distributions, $F$ and $H$, $\D_\alpha (F|H)$. 
We now define \emph{(pure) differential privacy} (PDP) \citep{Dwork2006} and \emph{zero-concentrated differential privacy} (zCDP) \cite{Bun2016}. 
\begin{definition}\label{dfn::pdp}
The quantity $\theta\sim G_{\bfx_n}$ is $\varepsilon$-differentially private ($\varepsilon$-PDP), $\varepsilon > 0$, if 
    \begin{equation*}
    \sup_{(\bfy_n,\bfz_n)\in \cA_n}\sup_{x\in \re}\frac{g_{\bfy_n}(x)}{g_{\bfz_n}(x)}\leq e^\varepsilon,
    \end{equation*}
and is $\rho$-zero-concentrated differentially private ($\rho$-zCDP), $\rho > 0$, if, for all $\alpha\in (1,\infty)$, 
\begin{equation*}
    \sup_{(\bfy_n,\bfz_n)\in \cA_n}\D_\alpha (G_{\bfy_n}|G_{\bfz_n})\leq \alpha\rho.
    \end{equation*}
\end{definition}
\noindent Here, $\theta$ is a $d$-dimensional differentially private quantity and $\varepsilon,\rho$ are the \emph{privacy budgets}, with smaller values of $\varepsilon,\rho$ producing higher levels of privacy. 
Quantities that are $\varepsilon$-PDP are also $\varepsilon^2/2$-zCDP. 
The following two properties are useful for developing differentially private estimates. 
\begin{proposition}\label{prop::dp} The following hold:
\begin{itemize}
    \item \textbf{Post-processing:} For any function $f$ that does not depend on $\bfx_n$, defined on a subset of $\rdd$, if $\theta$ satisfies $\varepsilon$-PDP ($\rho$-zCDP), then $f(\theta)$ satisfies $\varepsilon$-PDP ($\rho$-zCDP). 
    \item \textbf{Composition:} If $\theta_1,\theta_2$ satisfy $\varepsilon_1$-PDP and $\varepsilon_2$-PDP ($\rho_1$-zCDP and $\rho_2$-zCDP), respectively, then $(\theta_1,\theta_2)$ satisfies $(\varepsilon_1+\varepsilon_2)$-PDP ($(\rho_1+\rho_2)$-zCDP).
\end{itemize}
\end{proposition}
Proposition~\ref{prop::dp} first states that functions of differentially private quantities are still private, and second that releasing two differentially private quantities results in a jointly differentially private estimate, with a larger privacy parameter. 

Next, we introduce a simple way to generate a differentially private statistic: additive noise mechanisms. 
For a  univariate statistic $T\colon\cD_n\to \re$, define the global sensitivity of $T$ as 
$$\GS(T)=\sup_{(\bfy_n,\bfz_n)\in \cA_n}|T(\bfy_n)-T(\bfz_n)|.$$
We have the following proposition, introducing the Laplace and Gaussian mechanisms. 
\begin{proposition}\label{prop::adm}
Let $T\colon\cD_n\to \re$ be any statistic with $\GS(T)<\infty$. 
If $Z$ is a standard Laplace random variable, and $Z'$ is a standard Gaussian random variable, then 
$$T(\bfx_n)+\GS(T)Z/\varepsilon,\ \text{ and }\ T(\bfx_n)+\GS(T)Z'/\sqrt{2\rho},$$ satisfy $\varepsilon$-PDP and $\rho$-zCDP, respectively. 
\end{proposition}
\noindent Proposition~\ref{prop::adm} illustrates how correctly calibrated noise can be added to a statistic to ensure differential privacy. 

\subsection{Lugosi and Mendelson's modified winsorized mean} 
We now review the non-private inspiration for the proposed estimator. 
\citet{Lugosi2021} recently showed that, out of the class of non-private estimators, a modified winsorized mean estimator is minimax optimal, for the model given in Condition \ref{cond::problem}. 
Given a clipping proportion $0<p<1/2$, the idea is to first use half of the sample, $\bfY_n'=\{Y_i'\}_{i=1}^n$, to estimate the $p$th quantile and the $(1-p)$th quantile of $F$. 
We denote these quantile estimates by $\hat\xi_{p}$ and $\hat\xi_{1-p}$, respectively.
Then, we project the other half of the observations $\bfX_n'=\{X_i'\}_{i=1}^n$ onto the interval $[\hat\xi_{p},\hat\xi_{1-p}]$. 
The sample mean of the projected observations is taken as the final estimate of $\mu$. 
Formally, for $0<q\leq 1$, let $\hat\xi_q$ be the (left-continuous) $q$th quantile of the empirical distribution defined by $\bfY_n'$. 
For $a\leq b$, define
$$\phi_{a,b}(x)=\begin{cases}
    a & x<a\\
    x & a\leq x\leq b\\ 
    b & x>b
\end{cases}.$$
The non-private clipped mean is $$\hat\mu_p=\sum_{i=1}^n\frac{1}{n}\phi_{\hat\xi_{p},\hat\xi_{1-p}}(X_i').$$
\citet{Lugosi2021} proved an upper bound on $|\hat\mu_p-\mu|$ that holds with probability $1-\delta$, with $\delta$ at most exponentially small in $n$. 
Their result shows that taking $p=8\eta+12\log(4/\delta)/n$ yields an estimator that is minimax optimal. 
Here, we extend $\hat\mu_p$ to the differentially private setting. 
The motivation is that, for a carefully chosen $p$, the optimality in the nonprivate setting should translate to the private setting, provided we are able to estimate extreme quantiles well under the restriction of differential privacy. 

\subsection{Private quantile estimation}
A private version of the clipped mean, $\hat\mu_p$, requires differentially private estimation of extreme quantiles. 
We make use of the \emph{unbounded} private quantile estimator, developed by \citet{Durfee2024}. 
We use the unbounded quantile estimator as it has been shown to estimate extreme quantiles consistently, whereas other private quantile estimators do not \citep{Ramsay2024}. 
Specifically, when there is no contamination, our mean estimator requires a private analogue of an $O(1/n)$ quantile. 
It is expected that a private $O(1/n)$th quantile would be consistent for the infimum of the support of $F$, as $n\to\infty$. 
\citet{Ramsay2024} show that other private quantile estimators do not have this property, however, the unbounded estimator does.  

We present the modified estimator given by \citet{Ramsay2024}.
Let $V,V_1,V_2,\ldots$ be a sequence of independent, standard exponential random variables. 
Suppose that we wish to estimate the $q$th quantile of the population distribution privately, where the desired quantile is an upper quantile, (i.e., $1/2\leq q\leq 1$). 
First, we supply the algorithm with: a grid size parameter $\beta>1$, a lower bound on the desired quantile $\ell$, a dataset with empirical distribution $F_n$, and two privacy budgets $\varepsilon_1,\varepsilon_2>0$.
Given these parameters, the algorithm proceeds as follows: (i) Compute $\hat q=q+V/n\varepsilon_1$ and set $i=1$. (ii) Increase $i$ by one, stopping at the first $i$ such that $F_{n}(\beta^i+\ell -1)+V_i/n\varepsilon_2> \hat q$. 
(iii) Output $\beta^i+\ell -1$. 
If $0<q< 1/2$, we supply an upper bound on the desired quantile, $u$. Then, we negate the dataset, and run the upper-quantile algorithm, with input quantile $1-q$ and $\ell=-u$. 
This produces a negated estimate of the lower quantile that we negate to get the final quantile estimate.
We refer to these quantile estimates by $\tilde\xi_{q}\coloneqq\tilde\xi_{q,\bfX_n}(\varepsilon_1,\varepsilon_2,\ell,u,\beta) $. 
The algorithm is $(\varepsilon_1+\varepsilon_2)$-differentially private \citep{Durfee2024}. 

We show (Lemma~\ref{lem::PQE-zCDP}) that replacing $\varepsilon_i$ with $\sqrt{\rho_i}$ and taking $V,V_1,V_2,\ldots$ to be standard Gaussian random variables produces $(\rho_1+\rho_2)$-zero-concentrated differentially private quantiles. 
We denote the estimates generated by the zero-concentrated differentially private version of the algorithm by $\tilde \xi_{q}'\coloneqq\tilde\xi_{q,\bfX_n}'(\rho_1,\rho_2,\ell,u,\beta)$. 
As a by-product of our analysis, we prove a simple concentration result concerning both versions of these private quantiles (Lemmas~\ref{lem::pq-bound}-\ref{lem::pq-bound-zCDP}).

\section{A new differentially private mean estimator}\label{sec::main}
We now introduce our private mean estimator, the PMW mean, drawing inspiration from \citet{Lugosi2021} and \citet{Durfee2024}. 
Given a clipping proportion $0<p<1/2$ and private quantile parameters $\ell,u\in\re$, $\beta>1$, and $\varepsilon_1,\varepsilon_2>0$ (PDP) or $\rho_1,\rho_2>0$ (zCDP), we first compute $\tilde\xi_{p}$ and $\tilde\xi_{1-p}$ (PDP) or $\tilde\xi_{p}'$ and $\tilde\xi_{1-p}'$ (zCDP). 
Recall that these are the private estimates of the extreme quantiles, based on $\bfY_n'$. 
The data $\bfX_n'$ are then projected onto the interval $[\tilde\xi_{p},\tilde\xi_{1-p}]$ (PDP) or $[\tilde\xi_{p}',\tilde\xi_{1-p}']$ (zCDP). 
Lastly, we take the sample mean of the projected observations, and add calibrated noise via the Laplace mechanism (PDP) or the Gaussian mechanism (zCDP). 
Given $4e^{-n}<\delta<1$, our main result is a bound on the deviations of the PMW mean about $\mu$ that holds with probability $1-\delta$. 
Here, for two real numbers $a$ and $b$, we write $a\wedge b$ ($a\vee b$) to denote the minimum (maximum) of $a$ and $b$. 
Under two mild conditions, PMW mean is minimax optimal for the problem stated in Condition~\ref{cond::problem}, where the class of estimators is constrained to those that are differentially private (PDP or zCDP), if:
$$p=\zeta=16\eta+\frac{112}{3}\frac{\log(32(\beta(u-\ell) /(\beta-1)\vee 1)/\delta)}{n}.$$ 
For notational simplicity, we hide the dependence of $\zeta$ on other parameters, but formally, $\zeta\coloneqq\zeta(n,\eta,\delta,\ell,u,\beta)$. 
We formally define the estimators as follows.
\begin{definition}
Given private quantile parameters $\ell,u\in\re$, $\beta>1$, and $\varepsilon_1,\varepsilon_2,\varepsilon_3>0$ (PDP) or $\rho_1,\rho_2,\rho_3>0$ (zCDP), the PDP PMW mean is 
$$\tilde\mu=\sum_{i=1}^n\frac{1}{n}\phi_{\tilde\xi_{\zeta},\tilde\xi_{1-\zeta}}(X_i')+Z_1\frac{(\tilde\xi_{1-\zeta}-\tilde\xi_{\zeta})}{n\varepsilon_3},$$
and the zCDP mean $\mu'$ is defined analogously, replacing $\tilde\xi_.$, $\varepsilon_3$, and $Z_1$ with $\tilde\xi_.'$, $\sqrt{2\rho_3}$, and $Z_2$.
Here $Z_1$ is a standard Laplace random variable and $Z_2$ is a standard Gaussian random variable.   
\end{definition} 
\noindent It follows from basic composition (Proposition \ref{prop::dp}) that $\tilde\mu$ is $(2\varepsilon_1+2\varepsilon_2+\varepsilon_3)$-differentially private and that $\tilde\mu'$ is $(2\rho_1+2\rho_2+\rho_3)$-differentially private.  
Our main result requires two conditions on $\zeta$, $F$, and the input parameters of the private quantile algorithm $(\ell,u,\beta)$. 
\begin{condition}\label{cond::bounds}
$\ell\leq \xi_{5\zeta/4} $, $u\geq \xi_{1-5\zeta/4}$, and $\zeta<1/2$.
\end{condition}
Condition~\ref{cond::bounds} requires that the interval $(\ell,u)$ contains the interval $(\xi_{5\zeta/4}, \xi_{1-5\zeta/4})$, which, colloquially, means ``most of the uncontaminated data''. Our theoretical results show that the bounds $\ell\leq \xi_{5\zeta/4} $, $u\geq \xi_{1-5\zeta/4}$ can be loose, replaced by loose bounds on $\mu$ and $\sigma$, or replaced with a weaker condition, see Remark~\ref{rem::uell} and Remark~\ref{rem::c2}. 
\begin{condition}\label{cond::beta}
It holds that $\beta>1$ and that
\begin{equation*}
    \label{eqn::c2}
  \beta<1+\left[\frac{\xi_{5\zeta/4}-\xi_{3\zeta/4}}{u-\xi_{5\zeta/4}+1}\wedge  \frac{\xi_{1-3\zeta/4}-\xi_{1-5\zeta/4}}{\xi_{1-5\zeta/4}-\ell+1}\right]\coloneqq 1+b_n.
\end{equation*}
\end{condition}

Condition~\ref{cond::beta} says that the grid size used for quantile estimation cannot be too coarse.  
Combining this condition with our main result yields a range of $\beta$ values that produce minimax optimal estimates, up to logarithmic terms, see Remark~\ref{rem::beta}. 
We now introduce our main theoretical result.
For quantities $a,b$, we write $a\lesssim b$ ($a\gtrsim b$) when there is a universal constant $C>0$ such that $a\leq Cb$ ($a\geq Cb$). 
\begin{theorem}\label{thm::main-result}
Take $\eta\geq 0$, $4e^{-n}<\delta<1$, $\varepsilon_1,\varepsilon_2,\rho_1,\rho_2>3/56$,  and $\varepsilon_3,\rho_3>0$. For all $n\geq 1$, $u,\ell\in\re$, and $\beta>1$ such that Conditions \ref{cond::problem}--\ref{cond::beta} hold, with probability at least $1-\delta,$
\begin{equation}\label{eqn::main-result}
   |\tilde \mu-\mu|\lesssim    \sigma\sqrt{\frac{\log((u-\ell) \times \beta/(\beta-1)\times 4/\delta)}{n}}\\
  + \sigma\sqrt{\eta}+\frac{\sigma}{\sqrt{n}\varepsilon_3},
\end{equation}
and $|\tilde \mu'-\mu|$ is upper bounded (in the sense of $\lesssim$) by the same quantity, with $\varepsilon_3$ replaced with $\sqrt{\rho_3}$. 
\end{theorem}
The proof of Theorem~\ref{thm::main-result} is in Appendix \ref{sec::proof-main}. 
To prove Theorem~\ref{thm::main-result} we first prove a concentration inequality on the private quantiles. 
After which, it suffices to bound the quantity $n^{-1}\sum_{i=1}^n\phi_{\xi_{\zeta},\xi_{1-\zeta}}(X_i')-\mu$ with high probability. 
For this task, the proof is similar to the proof of Theorem 1 of \citet{Lugosi2021}. 
However, their proof contains an error, that is resolved here (Appendix~\ref{app::lm-error}). 
A sample complexity analogue of Theorem \ref{thm::main-result} is given in Appendix~\ref{app::SC}.

Interpreting Theorem~\ref{thm::main-result}, we first observe that, provided that $\log(\beta(u-\ell)/(\beta-1))\lesssim \log(4/\delta)$, and Conditions~\ref{cond::bounds} and \ref{cond::beta} are satisfied, the rate given in Theorem~\ref{thm::main-result} is minimax optimal, regardless of the level of contamination. 
The lower bound ensuring minimax optimality can be generated by combining lower bounds from private \citep{Kamath2020} and nonprivate \citep{Lugosi2019} mean estimation.

\begin{remark}[On choosing $\beta$]\label{rem::beta}
For minimax optimality, we require $\log(\beta/(\beta-1))\lesssim \log(4/\delta)$. 
Combining this with Condition~\ref{cond::beta} yields a recommended range for the grid size: $(1-(\delta/4)^s)^{-1}\leq 1+(\delta/4)^s/(1-(\delta/4)^s) \leq\beta\leq 1+b_n$, for some $s>0$. 
For example, we can select $\beta=O(1+n^{-s}/(1-n^{-s}))$, provided, for some $s>1$, $\delta<n^{-s}<b_n$. 
This ensures that the bound \eqref{eqn::main-result} holds with probability at least $1-n^{-s}$. 
Regarding $b_n$, population densities which have a bounded support have smaller values of $b_n$, when compared to ones with an unbounded support. 
For instance, for the standard uniform distribution, $b_n=O(1/n)$, for the standard exponential distribution, $b_n\gtrsim 1/(u-\log n+1)$, and for Student's $t$-distribution with three degrees of freedom, $b_n\gtrsim 1,$ when $u=2\xi_{1-5\zeta/4}$ and $\ell=-u$. 
Thus, we must take a finer grid when applying the procedure to distributions with bounded support, but a coarse grid is acceptable for distributions with unbounded support. 
\end{remark}
\begin{remark}[On choosing $u$ and $\ell$]\label{rem::uell}
Condition~\ref{cond::bounds} requires that $\ell<\xi_{5\zeta/4}<\xi_{1-5\zeta/4}<u$. 
Then, \eqref{eqn::main-result} implies that taking $\log(u-\ell)\lesssim\log(1/\delta)$ produces a PMW mean that is minimax optimal. 
Choosing $\ell,u$ thus amounts to ensuring these two bounds are satisfied. 
First, observe that, in general, we take $\delta$ to be small, say, $\delta=n^{-k}$ for some $k>0$. 
The bound $\log(u-\ell)\lesssim\log(1/\delta)$ then requires that $u-\ell$ is at most polynomial in $n$. 
Now, we need $u,\ell$ such that $\ell<\xi_{5\zeta/4}<\xi_{1-5\zeta/4}<u$. 
For $F$ with bounded support, this is straightforward. 
If the support is unbounded, we can specify $u,\ell$ according to the following argument. 
We have that, by Chebyshev's inequality, $5\zeta/4\leq \sigma^2/(\xi_{1-5\zeta/4}-\mu)^2 $. 
This implies that $\xi_{1-5\zeta/4}\lesssim \sigma/\zeta^{1/2}+\mu$. 
Therefore, Condition~\ref{cond::bounds} is implied by $u\gtrsim \sigma/\zeta^{1/2}\vee \mu$. 
Then, using $\delta=n^{-k}$, $\sigma/\zeta^{1/2}\lesssim\sigma  \sqrt{n/k\log n}.$
Assuming that $\mu,\sigma$ are constant in $n$, we take $u\geq  n^{1/2+r},$ $\ell\leq -n^{1/2+r},$ for some $r>0$.
For large $n$, this choice will imply Condition~\ref{cond::bounds}. 
If one has instead (potentially very loose) bounds such that $-\mu_0<\mu<\mu_0$ and $\sigma<\sigma_0$, then taking $u\gtrsim  \sqrt{n\sigma_0}\vee \mu_0$ and $\ell\lesssim  -(\sqrt{n\sigma_0}\vee \mu_0)$ implies Condition~\ref{cond::bounds}, and ensures that $u-\ell$ is at most polynomial in $n$. 
\end{remark}
\begin{remark}[Weakening Condition~\ref{cond::bounds}]\label{rem::c2}
Condition~\ref{cond::bounds} requires that the interval $(\ell,u)$ contains most of the uncontaminated data. 
This condition can be relaxed to only require that $u>\xi_{5\zeta/4}$ and $\ell<\xi_{1-5\zeta/4}$, (i.e., $u$ is an upper bound on the lower quantile and $\ell$ is a lower bound on the upper quantile). 
This will come at the expense of replacing $u-\ell$ in \eqref{eqn::main-result} with $\log n$. 
This modification still constitutes a minimax optimal rate, up to logarithmic terms. 
However, given that the bounds on the data can be very loose, as per \eqref{eqn::main-result}, and the results in Section~\ref{sec::simulation2}, we leave Condition~\ref{cond::bounds} unchanged. 
\end{remark}
\begin{remark}[Subgaussian distributions and more]
As usual, (e.g., as in the work of \citet{Lugosi2021}), we can derive stronger bounds under stronger assumptions. 
We use Chebyshev's inequality to generate the first and last terms in \eqref{eqn::main-result}. 
This can be replaced with the tightest tail bound given the assumptions. 
For instance, if we assume the uncontaminated data are subgaussian, we use a subgaussian tail bound to get tighter deviation bounds. 
In that case, the PMW mean also achieves optimal rates. 
\end{remark}
\begin{remark}[Multivariate case]
The PMW mean can be extended in a coordinate-wise fashion to give a simple multivariate estimator. 
Given the coordinate-wise nature, this estimator will not achieve minimax optimality. 
For instance, even under no contamination, the coordinate-wise estimator has error at least $d^{3/2}/\sqrt{n}\epsilon$, which is sub-optimal \citep{Kamath2020}. 
Another way of looking at this is to see that the coordinate-wise estimator has error at least $\sqrt{d\eta}$. 
There are (non-private) estimators (e.g., in \citet{Lugosi2021}) which depend on the contamination only through $\sqrt{\eta\lambda_1}$, where $\lambda_1$ is the largest eigenvalue of the covariance matrix. 
It stands to reason that we should be able to do better than $\sqrt{d\eta}$ in the private setting, since we are often adding noise in a clever way to non-private estimators.
It may be possible to extend our estimator to a more efficient multivariate estimator, say a private version of the multivariate estimator of \citet{Lugosi2019}. 
However, even non-privately, this is a difficult computational problem \citep{Lugosi2019}. 
The goal of this work was to develop a practical and fast algorithm to use in subsample-and-aggregate, not to develop a minimax optimal multivariate estimator. 
As such, we leave extensions of the proposed estimator to the multivariate setting for future work. 
\end{remark}
\section{A more practical private clipped mean}\label{app::more_practical}
Next, we present adjustments to the proposed estimator $\tilde\mu'$ that we recommend in practice, as well as a proposed practical version of $\tilde\mu'$, denoted by $\tilde\mu'_p$. 
First, note that splitting the sample in half to compute each of the two steps of the estimator does not have any algorithmic or intuitive justification. 
It is an arbitrary choice that simplifies the proof of the deviation bound.
Therefore, when computing $\tilde\mu'_p$, we modify the algorithm and use all data in each estimation step. 
Note that this does not change the privacy guarantees of the estimator.

Next, observe that no effort was made to optimize the constants in the clipping parameter $\zeta$ in the theoretical exposition of Section \ref{sec::main}.
Therefore, using $\zeta$ as a guide, we present the following practical version of the clipping parameter, which will be used in simulation for both $\tilde\mu'_p$ and $\tilde\mu'$. 
The user is asked to select two parameters, $0<C<0.025 n$ and $0\leq \eta<1/2$. 
As before, $\eta$ is the maximum expected level of contamination.
The new parameter, $C$, represents the number of observations the user would like to clip from either end, if the contamination level was 0. 
We cap $C$ at $0.025n$ as otherwise, at very small sample sizes, $C$ may be larger than $n/2$.
Finally, the clipping parameter is computed to be $C/n \vee \eta$, and is used in place of $\zeta$. 
We use a simulation study to investigate different values of $C$ and $\eta$, and found that the value of $C$ was not particularly important (see Appendix \ref{app::crho}).
In addition, as expected, we found that smaller $\eta$ was better for uncontaminated, light-tailed, uncontaminated data (Gaussian, Exponential and Gaussian mixtures) but larger $\eta$ was preferable for heavy tailed or contaminated data (Contaminated Gaussian and student $t$) (see Appendix \ref{app::crho}).

Next, not covered by the present theory, is how to best split the total budget $\rho$ amongst $\rho_1,\rho_2,$ and $\rho_3$. 
For this, we simulated the performance of the estimator under three scenarios: more weight to the quantile estimation, equal weight to the mean estimation and quantile estimation, and more weight to the mean estimation (see Appendix \ref{app::pbudg}). 
We found that this was not particularly important, except at very small samples, and so we simply split the budget evenly between the mean estimation and quantile estimation, specifically setting $\rho_1=\rho_2=\rho/4$ and, $\rho_3=\rho/2$ in the following simulation results. 

\section{Simulation Study}\label{sec::simulation2}
\subsection{Simulation study settings}\label{app::sim_settings}
For each run of the simulation for $\tilde\mu'$ and $\tilde\mu'_p$, we set $\eta=0$ or $\eta=0.3$, take $C$ (as described in Appendix~\ref{app::more_practical}) to be uniformly, randomly chosen between 1 and 100, and set $\beta=1.001$ (as in \citep{Durfee2024}), $\rho_1=\rho_2=\rho/4$, $\rho_3=\rho/2$, and $\ell=-50<u=50$. 
For the existing estimators, any estimator that required input bounds had input bounds set at $\ell=-50<u=50$, and for any estimator that required setting number of moments, we set the number of moments at $2$. 
We compared our algorithm to the estimators of \citet{Barber2014}, denoted \texttt{Duchi} with $r=25$, \citet{Kamath2020} denoted \texttt{Kamath} with $\alpha=1$, $R=u=50$, and $\beta=0.1$, and the Laplace Log-Normal mean algorithm of \citet{Bun2019} denoted \texttt{Bun} with $m=C$ and a grid of values for $t$ ranging from 0.01 to 1. 
In addition, for \texttt{Kamath}, we set $m=5\log 20$ instead of $m=200\log20$ to accommodate the small sample sizes. 
We consider the performance of the estimators under five different population distributions (Gaussian, Gaussian mixture, Exponential, heavy tailed Student's $t$, and a 20\%-contaminated Gaussian) under varying privacy budgets $\rho\in \{0.1,0.5,1\}$ and sample sizes $n\in\{50,100,500,1000\}$. 
Each combination of the above scenarios was run 250 times.
\begin{figure*}
    \centering
\includegraphics[width=1\textwidth]{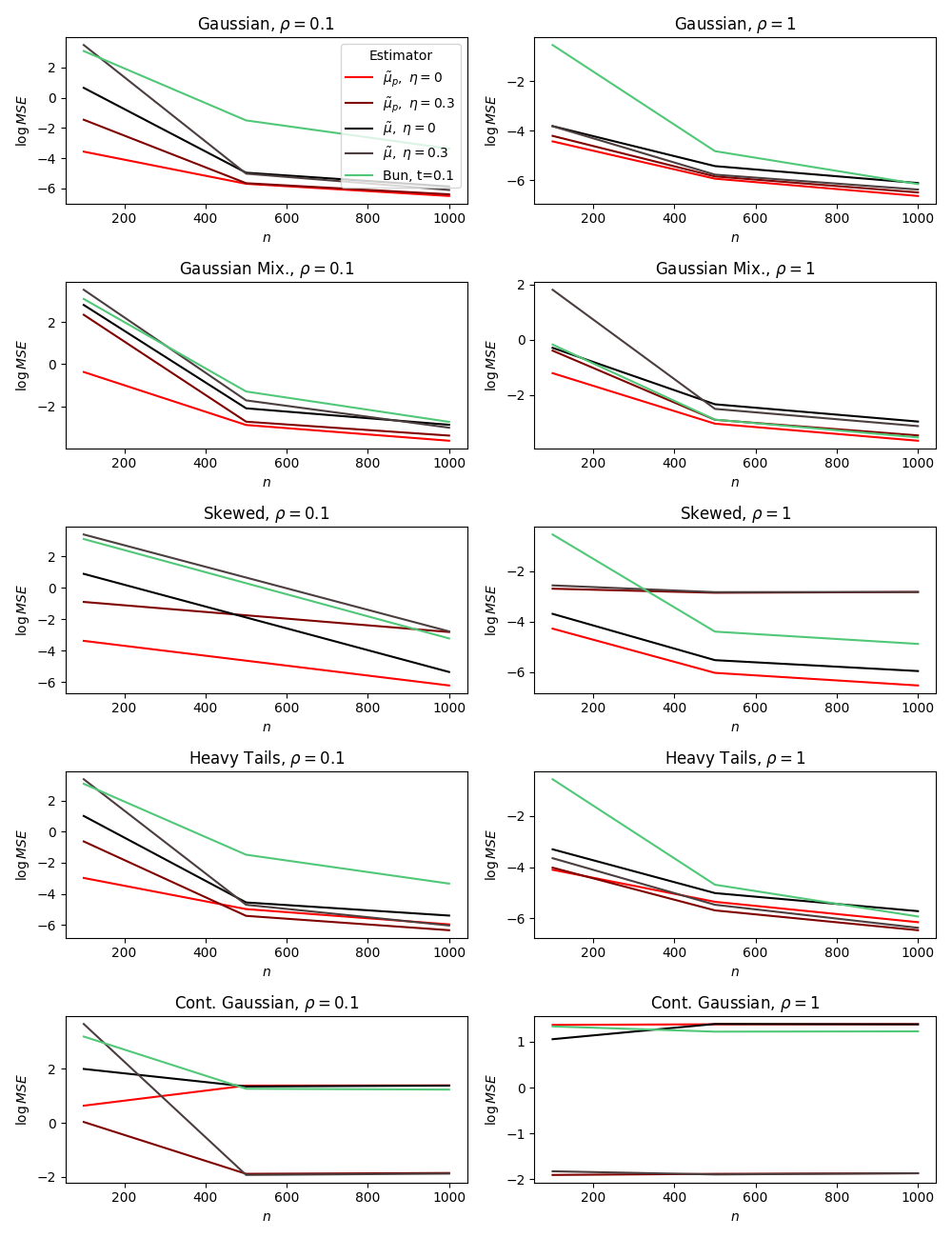}
    \caption{The empirical $\log MSE$ of $\tilde\mu_p,\ \tilde\mu$, with $\eta=0$ and $\eta=0.3$, compared to the estimator of \citet{Bun2019}, Bun, with $t=0.1$, compared across a variety of different population distributions (Gaussian, mixture of Gaussians, Skewed, Heavy-Tailed, and a Contaminated Gaussian) and different privacy budgets ($\rho$).}
    \label{fig:sim-results}
\end{figure*}
\begin{table*}[t!]
\centering
\caption{Empirical variance of the additive noise in the trimmed mean, \texttt{Bun}, with $t=0.1$, and the PMW mean, $\tilde\mu_p'$, with $\eta=0$, compared across various population distributions, over 250 runs. Here, $C$ and $m$ (as defined in Appendix~\ref{app::more_practical} and \citep{Bun2019}, respectively) were randomly selected between 1 and 100.}
\vspace{0.7em}
\resizebox{\linewidth}{!}{\begin{tabular}{ll|rr|rr|rr}
& & \multicolumn{2}{c}{Gaussian} & \multicolumn{2}{c}{Gaussian Mixture} &  \multicolumn{2}{c}{Skewed}\\
$\rho$ & $n$  & \texttt{Bun} & $\tilde\mu_p',\ \eta=0$ & \texttt{Bun} & $\tilde\mu_p',\ \eta=0$ & \texttt{Bun} & $\tilde\mu_p',\ \eta=0$ \\
 \hline
\multirow[c]{4}{*}{0.1} & 50 & 1.48e+02 & \textbf{1.82e+00} & 1.48e+02 & \textbf{8.45e+00} & 1.48e+02 & \textbf{2.54e+00} \\
 & 100 & 2.20e+01 & \textbf{1.03e-02} & 2.20e+01 & \textbf{1.42e-01} & 2.20e+01 & \textbf{8.00e-03} \\
 & 500 & 2.21e-01 &\textbf{ 7.22e-04} & 2.26e-01 & \textbf{5.86e-03} & 2.22e-01 & \textbf{6.05e-04} \\
 & 1000 & 3.26e-02 & \textbf{1.95e-04} & 3.67e-02 &\textbf{ 1.85e-03 }& 3.27e-02 & \textbf{1.70e-04} \\
 \hline
\multirow[c]{4}{*}{1} & 50 & 3.87e+00 & \textbf{4.99e-03} & 3.87e+00 & 5.98e-02 & 3.87e+00 & \textbf{3.38e-03} \\
 & 100 & 5.65e-01 & \textbf{1.54e-03} & 5.65e-01 & \textbf{1.63e-02} & 5.65e-01 & \textbf{9.81e-04} \\
 & 500 & 5.35e-03 & \textbf{7.81e-05} & 5.50e-03 & \textbf{7.50e-04} & 5.36e-03 & \textbf{5.64e-05} \\
 & 1000 & 7.36e-04 & \textbf{2.06e-05} & 8.48e-04 & \textbf{1.88e-04} & 7.39e-04 & \textbf{1.48e-05} \\\hline
\end{tabular}}
\label{tab:noise-var}
\end{table*}
Next, we compare the empirical performance $\tilde\mu'$ and  $\tilde\mu_p'$, (described in Section~\ref{app::more_practical}), to several existing estimators. 
For brevity, we only consider zCDP estimators. 

Several of the results are omitted from the main figure to save space. The complete set of results are available in Appendix~\ref{app::xtra-sims}.
For the main results, we note that \texttt{Duchi} and \texttt{Kamath} did not perform as well as \texttt{Bun}, $\tilde\mu'$ and $\tilde\mu'_p$, so we omit these results. 
Additionally, for \texttt{Bun}, $t=0.1$ performed the best overall, see Appendix~\ref{app::xtra-sims}, so we omit the results with other values of $t$ from the main body. 
Lastly, the results from $\rho=0.5$ are essentially an interpolation between those of $\rho=0.1$ and $\rho=1$, and so we omit those from the main body. 

Figure~\ref{fig:sim-results} displays the log-transformed empirical mean squared errors of each estimator, for each scenario. 
First, $\tilde\mu_p'$ outperforms $\tilde\mu'$, and so we focus on comparing $\tilde\mu_p'$ to \texttt{Bun}. 
Now, observe that except in the contaminated Gaussian case, $\tilde\mu_p'$ with $\eta=0$ outperforms \texttt{Bun}. 
It is even the case that $\tilde\mu_p'$ with $\eta=0.3$ outperforms \texttt{Bun} at smaller sample sizes. 
This can be explained by the additive noise mechanism used; we examine the empirical variance of the additive noise for each estimator (Table~\ref{tab:noise-var}). 
Table~\ref{tab:noise-var} shows that the average magnitude of noise added to construct $\tilde\mu_p$ with $\eta=0$ is much smaller than that used to construct \texttt{Bun}, especially at small sample sizes. 
To explain the exception of the contaminated Gaussian, observe that when the extreme observations are contaminants, it is better to remove them, rather than project them. 
This explains the superior performance of \texttt{Bun} compared to $\tilde\mu_p'$ with $\eta=0$ in the contaminated Gaussian. 

Let us now consider what happens when we suspect contamination and set $\eta=0.3$. Observe that $\tilde\mu_p'$ with $\eta=0.3$ sacrifices some efficiency in the uncontaminated case for robustness in the heavy tailed and contaminated Gaussian case. Notable is the low privacy, skewed case. 
For skewed data, clipping a constant fraction of data when there is no contamination, (i.e., $\eta=0.3$) results in higher error estimates. 
This is because clipping a constant fraction of observations will result in an inconsistent estimator, non-privately. 
A similar phenomenon occurs with the \texttt{Bun} estimator. 
Observations in the algorithm of \texttt{Bun} are removed, rather than projected. 
Removing an equal number of observations from each extreme end of the dataset is predicated upon an implicit symmetry assumption. 
When the distribution is skewed, and we remove $m$ observations from either end of the sample, we are trimming a tail on one side and not the other, resulting in a biased estimator. 
Theoretically, the (non-private) trimmed mean and clipped mean are both inconsistent whenever $m/n\to p>0$ (see Lemma~\ref{lem::tm-inc}). 

\section{Improved subsample-and-aggregate}\label{sec::data-analysis}

We now demonstrate that the proposed mean estimator can be used to improve subsample-and-aggregate. 
Subsample-and-aggregate, developed by \citet{Nissim2007}, and extended by \citet{smith2011privacy}, is a method for estimating virtually any quantity, privately. 
Define a multivariate dataset to be a set of $N=2n$, $p$-dimensional vectors, and let $\cD_{N,p}$ be the set of all such datasets. 
Consider a multivariate dataset, say $\bfX_N$, a set of $N$ realizations of $p$-dimensional random vectors, that are independent and identically distributed. 
To compute a statistic $T\colon \cD_{N,p}\to \rdd$ for $d\geq 1$, from $\bfX_N$, privately, we:
(i) Split the dataset $\bfX_N$ into $m$ disjoint groups of size $k=\floor{N/m}$, say $\bfX_N^{(1)},\ldots,\bfX_N^{(m)}$.
(ii) Compute $T_1,\ldots,T_m$, where $T_i=T(\bfX_N^{(i)})$.
(iii) Apply a private mean estimator to the set of subsampled estimates $T_1,\ldots,T_m$ to generate a private estimate of $T(\bfX_N)$. 
As mentioned previously, the private aggregator in subsample-and-aggregate algorithm is typically either the traditional private clipped mean, or the private widened winsorized mean applied coordinate-wise \citep{Nissim2007, smith2011privacy}. 
If poor input bounds are provided for the coordinates of $T(\bfX_N)$, this approach introduces a high level of noise to the private estimate. 
Moreover, we will demonstrate that other popular private mean algorithms can be unstable, or have large errors, when applied to small sample sizes and the dimension is not very large. 
Critically, in the context of subsample-and-aggregate, we are often computing the mean of a small sample, given that we must first split the data into groups large enough for $T_1,\ldots,T_m$ to be stable. 
This motivates our proposal of $\tilde \mu_p$ or $\tilde \mu'_p$ applied coordinatewise as the aggregator. 
We demonstrate this observation below, via a real-world application to linear mixed modelling.

First, we state a corollary of Theorem~\ref{thm::main-result}, which yields the finite sample accuracy of the private estimate generated from subsample-and-aggregate. 
Assume that $\bfX_N'$ is an adversarially corrupted version of $\bfX_N$, where at most $\floor{n\eta}$ observations are corrupted. 
Let $\mathrm{Bias}_{N,\eta}(T)=\E(T(\bfX_N'))-\theta$ be the bias of the estimator $T$. 
Next, given $T\colon \cD_{N}\to \re$, let $\tilde T$ and $\tilde T'$ be the subsample-and-aggregate estimates generated when $\tilde\mu$ and $\tilde\mu'$ are used for the aggregation step, respectively. 
Let $\bfX_N'^{(1)},\ldots,\bfX_N'^{(m)}$ be a uniformly random set of disjoint subsets of size $k=\floor{N/m}$ chosen from $\bfX_N'$, and let $G_k$ be the distribution of $T(\bfX_N'^{(1)})$. 
Lastly, define 
\begin{equation*}
    h(m,\eta,\delta,\beta,u,\ell,\varepsilon_3)=
    \sqrt{\eta}\vee \sqrt{\frac{(\log\left(\left[\frac{\beta(u-\ell)}{\beta-1}\right]\vee 4\right)/\delta)}{m}}\vee\frac{1}{\sqrt{m}\varepsilon_3}.
\end{equation*}
For a positive integer $m$, let $[m]=\{1,\ldots,m\}$.
\begin{corollary}\label{corr:ssa}
Given $T\colon \cD_{N}\to \re$ and $m\in[N]$, suppose that $T(\bfX_N'^{(1)})$ has variance $\sigma_{T,k}^2<\infty$, and $\ell<\inf_{i\in[m]}T(\bfX_N'^{(i)})<\sup_{i\in[m]}T(\bfX_N'^{(i)})<u$, and that for all $\eta\geq 0$, $4e^{-N}<\delta<1$, $\varepsilon_1,\varepsilon_2,\rho_1,\rho_2>3/56$, and $\varepsilon_3,\rho_3>0$, if $\beta,u,\ell\in\re$ are such that Condition \ref{cond::beta} holds for $F=G_k$, then, with probability at least $1-\delta,$
\begin{equation*}
       |\tilde T(\bfX_N')-T(F)|\lesssim |\mathrm{Bias}_{N,\eta}(T)|
       +  \sigma_{T,k}h(m,\eta,\delta,\beta,u,\ell,\varepsilon_3),
\end{equation*}
and $|\tilde T(\bfX_N')'-T(F)|$ is upper bounded (in the sense of $\lesssim$) by the sample quantity, but with $\varepsilon_3$ replaced with $\sqrt{\rho_3}$. 
\end{corollary}
Note that we do not necessarily need the restriction $$\ell<\inf_{i\in[m]}T(\bfX_N'^{(i)})<\sup_{i\in[m]}T(\bfX_N'^{(i)})<u,$$ as this is stronger than what is required by Condition~\ref{cond::bounds}, especially in light of Remark~\ref{rem::c2}. 
However, this condition makes the theorem easier to read and, since the bounds can be loose in practice, this restriction is not a practical concern.

Corollary~\ref{corr:ssa} yields two important insights. 
First, the optimal choice of $k$ depends on the bias of $T$. 
Second, under contamination, the rate of convergence of $\tilde T$ depends on the robustness of $T$. 
Suppose that $\eta=0$.
Taking the privacy parameters, $\ell$, $u$, and $\beta$ to be fixed, the rate of convergence is determined by two competing parameters, $|\mathrm{Bias}_{N,\eta}(T)|$ and $\sigma_{T,k}/m^{1/2}$. 
Provided that $\sigma_{T,k}=O(k^{-1/2})$, then $\sigma_{T,k}/m^{1/2}=O(N^{-1/2})$, which is the minimax optimal rate of convergence for many univariate problems. 

By choosing a statistic $T$ that has a variance $O(k^{-1/2})$, the error of subsample-and-aggregate is bounded by $|\mathrm{Bias}_{N,\eta}(T)|+ O(n^{-1/2})$. 
It remains then to select $k$ large enough so that $|\mathrm{Bias}_{N,\eta}(T)|=O(N^{-1/2})$. 
If $T$ is unbiased, then any choice of $k$ is sufficient for a convergence rate of $O(N^{-1/2})$. 
If $T$ is biased, then $k$ has to be large enough so that the bias is negligible. 
For instance, take $F=\cN(0,\sigma^2)$ and $T$ to be the maximum likelihood estimator of the variance of the normal distribution, $|\mathrm{Bias}_{N,\eta}(T)|=\sigma^2/k$ and $\sigma_{T,k}\lesssim 1/\sqrt{k}$. 
In this case, taking $k\gtrsim \sqrt{N}$ results in a convergence rate of $O(N^{-1/2})$, which is the optimal convergence rate in the nonprivate setting. 

For the second point, aside from the $\sqrt{\eta}$ term resulting from the mean estimation, note that the level of contamination is reflected in a finite-sample verion of the \textit{maximum bias} of the estimator $T$. 
The finite-sample maximum bias at level $\eta$ for $N$ samples is the maximum distance $T_i(\bfx_N')$ can be moved by corrupting $\eta$-fraction of the points in the dataset $\bfx_N'$, \citep{guide_huber1981}. 
Namely, if the sample is contaminated, in general, $|\mathrm{Bias}_{N,\eta}(T)|\leq |\E(T_i(\bfx_k'))-\theta|+\sup_{\bfx_k}|\E(T_i(\bfx_k'))-\E(T_i(\bfx_k))|$. 
This is the sum of the bias of $T$ in the uncontaminated setting and the contamination bias of $T$ for the dataset $\bfx_k'$. 
If contamination is suspected, it is necessary to select a statistic $T$ that is ``robust with high probability,'' in the sense that for most datasets, $T$ is not easily contaminated. 
This is different from the stricter requirement of small global sensitivity, where it is required that $\sup_{\bfx_k,\bfx_k'}|T_i(\bfx_k')-T_i(\bfx_k)|$ is bounded. 
For example, the univariate median has infinite global sensitivity, but will generally be robust with high probability \citep{Dwork2009, Avella-Medina2019a, 2023arXiv231207792R}.

We now turn to an application of the improved subsample-and-aggregate. 
We consider the data of 2327 healthcare workers, whose K10 combined depression and anxiety scores were tracked throughout the course of the COVID-19 pandemic \citep{Gutmanis2024}. 
One of the aims of the study was to estimate the effects of time, season, and periods of high mitigation on K10 scores. 
Given that this is longitudinal data, the authors applied a linear mixed effects model, specifically a random intercept model, to the log-transformed K10 scores. 
We estimate the same model privately, using the proposed improved subsample-and-aggregate. 
Specifically, we estimate the coefficients of time, winter, and mitigation level, as well as the intercept, the between subject variance, and the within subject variance. 
We further compare the output to the original subsample-and-aggregate, and subsample-and-aggregate algorithms where the aggregate step uses one of several state-of-the-art private multivariate mean estimation methods: PRIME \cite{2021Liub}, the private Huber estimator \citep{Yu2024}, or the Coinpress estimator \citep{Biswas2020}. 

Table~\ref{tab:diff-k-2} shows the empirical mean squared error between each of the private estimates and the non-private estimates, for group size $k=40$ with $\rho=1$. Note that $k=40$ corresponds to an aggregator sample size of roughly 60 observations (See Table~\ref{tab:diff-k} in Appendix~\ref{app::xtra-sims} for other values of $k$.)
Here, the provided input bounds for all algorithms were $\ell=- 40\sqrt{6}=-u$. 
The full parameter settings are provided in Appendix \ref{app::lmm-settings}. 
In these results $\tilde\mu_p'$ outperforms the original subsample-and-aggregate, and all of the multivariate methods by several orders of magnitude. 
In addition, several multivariate methods fail for this application. 
One explanation is that, to achieve minimax optimality in the preceding settings, one must take into account the dependence structure of the data. 
For instance, many of the existing multivariate estimators privately estimate the covariance matrix \citep{Biswas2020, Brown2023,Duchi2023}.
While these estimators generally perform very well, at small sample sizes this estimate of the covariance matrix can become unstable, leading to unstable mean estimates.

\begin{table*}[t]
\centering
\caption{Empirical mean squared error between the non-private mixed model parameter estimates and the mixed model parameter estimates produced by pairing subsample-and-aggregate algorithms with different private mean estimation algorithms, for $k=40$ at $\rho=1$.}
\resizebox{\linewidth}{!}{
\begin{tabular}{ccccc}
\toprule
 $\tilde\mu_p'$ & Clipped mean & PRIME \cite{2021Liub} & Huber \citep{Yu2024} & Coinpress \citep{Biswas2020}  \\
\midrule
 \textbf{0.235} &        28.901 & 2302.797 &  3.224 &   8906.546 \\
\bottomrule
\end{tabular}}
\label{tab:diff-k-2}
\end{table*}

\begin{figure}[t]
\centering
\includegraphics[width=\textwidth]{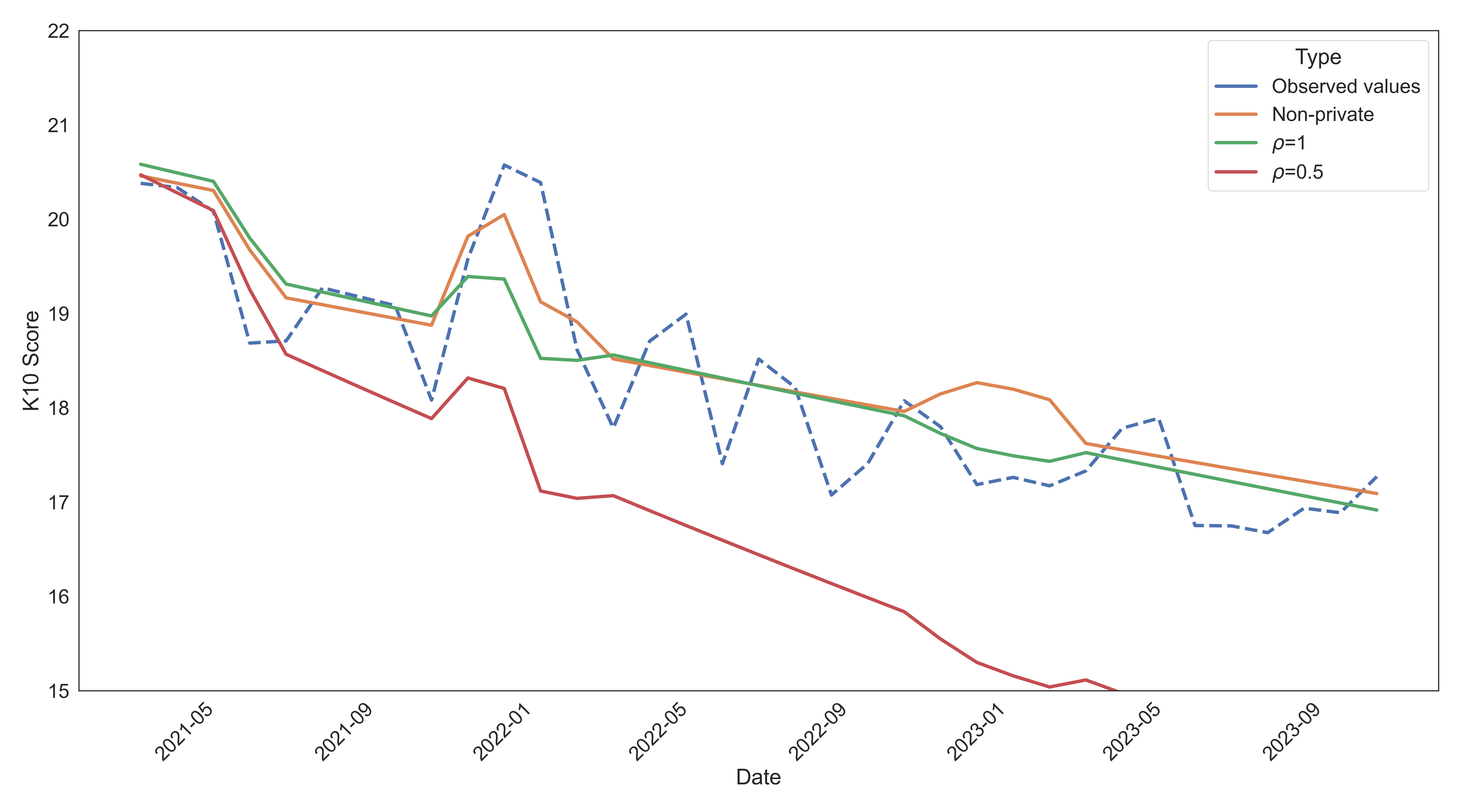}
\caption{Actual mean values, non-private estimated mean values and private estimated mean values of K10 scores of healthcare workers throughout the COVID-19 pandemic.}
\label{fig:output}
\end{figure}
\begin{table}[t]
\centering
\caption{Estimated percentage change in K10 scores for a unit increase in each regressor, with the estimated between ($\sigma^2_B$) and within ($\sigma^2_W$) subject variances,  across values of $\rho$, including non-private estimates ($\rho=\infty$). M.L. refers to mitigation level. In this case, we see that the results from $\rho=0.1$ are not helpful, but $\rho=1$ are fairly similar to their non-private counterparts. }
\label{tab:my_label}
\begin{tabular}{l|cccc}
& \multicolumn{3}{c}{$\rho$}\\
\hline
 & $\infty$  & $1$ & $0.5$ & $0.1$ \\
\midrule
Intercept            &         2.91 &      2.92 &        2.90 &          -2.17 \\
Winter               &         2.48 &     -1.06 &       -1.57 &         -78.46 \\
Time                 &       -12.18 &    -13.95 &      -27.25 &        -100.00 \\
Mitigation           &         5.13 &      4.71 &        6.20 & $39 \times 10^9$ \\
$\sigma^2_B$ & 0.09 &      0.08 &        0.08 &         -10.28 \\\
$\sigma^2_W$ &  0.04 &      0.03 &        0.04 &          34.44 \\
\end{tabular}
\end{table}

Figure~\ref{fig:output} and Table~\ref{tab:my_label} display the results from an actual private analysis done via the improved subsample-and-aggregate with $k=40$. 
It is helpful to note that we are effectively computing the mean of 60 observations, a fairly small sample. 
Figure~\ref{fig:output} displays the estimated mean K10 scores using the private and non-private estimates. The rolling, four-week means of the raw data are also included. 
Table~\ref{tab:my_label} compares the private parameter estimates generated from the improved subsample-and-aggregate, with the non-private estimates. 
Observe that the procedure works fairly well for $\rho=1$, but not as well for $\rho=0.5$, and performs unacceptably for $\rho=0.1$. 
It is expected that the results degrade as $\rho$ decreases, especially when sample sizes are small. In order to achieve strong privacy guarantees, regardless of the estimator being used, the analyst requires a sufficient amount of data. Despite the unacceptable performance of the proposed estimator with $\rho = 0.1$, it is worth re-emphasizing the improvement in these results when compared with other aggregators. At the present, this analysis explores the limits of differential privacy in practice, an area that is worthy of future investigation.

When considering the results with $\rho = 1$, the estimates align closely with the non-private estimates, except we estimate a reduction in the scores in the winter, when in fact, there is an increase. 
Otherwise, the fitted mean functions align well with the data. 
We conclude that, if the researchers had used a differentially private approach with $\rho=1$, their conclusions for the winter effect would differ, but the remaining conclusions would not change materially. 

In general, we stress that subsample-and-aggregate is a powerful algorithm. 
Complex algorithms and models can be made private easily, and with minimal additional programming. 
By removing the need for tight input bounds on the data, or the need for large sample sizes, we unlock private versions of virtually any algorithm, without the need to design a new procedure for every different algorithm we may wish to privatize. 




\section*{Acknowledgements}
We would like to acknowledge \citet{Yu2024} for providing their code. We would also like to acknowledge the authors of \citet{Biswas2020,2021Liub} for making their code available online.  
\par
	

\bibliographystyle{apalike}      
\bibliography{main_a}   
\appendix
\section{Sample complexity version of main result}\label{app::SC}
We now state the sample complexity analogue of Theorem \ref{thm::main-result}, when $\eta=0$, which may be of interest to some readers. 
Note, a version of Corollary \ref{cor::sc} at positive corruption values $\eta$ that are constant in $n$ is impossible, since the minimax lower bound stipulates that no estimator can be consistent, at least without imposing additional restrictions on $F$.
\begin{corollary}\label{cor::sc}
Given $4e^{-n}<\delta<1$, and $\varepsilon_1,\varepsilon_2,\varepsilon_3,\rho_1,\rho_2,\rho_3>0$, if the conditions of Theorem~\ref{thm::main-result} hold and $\eta=0$, then there exists a universal constant $K>0$ such that for all $t>0$, $|\tilde \mu-\mu|\leq t$ with probability at least $1-\delta,$ provided
\begin{equation}\label{eqn::main-result3}
   n\geq K\left(    \sigma^2 \frac{\log(u-\ell)-\log((\beta-1)/\beta)+\log(4/\delta)}{t^2}\vee\frac{1}{t^2\varepsilon_3}\right)
\end{equation}
and $|\tilde \mu'-\mu|\leq t$ with probability at least $1-\delta,$ provided
\begin{equation}\label{eqn::main-result4}
      n\geq K\left(    \sigma^2 \frac{\log(u-\ell)-\log((\beta-1)/\beta)+\log(4/\delta)}{t^2}\vee\frac{1}{t^2\sqrt{\rho_3}}\right).
\end{equation}
\end{corollary}
Note, a version of Corollary \ref{cor::sc} at positive corruption values $\eta$ that are constant in $n$ is impossible, since the minimax lower bound stipulates that no estimator can be consistent, at least without imposing additional restrictions on $F$.

\lhead[\footnotesize\thepage\fancyplain{}\leftmark]{}\rhead[]{\fancyplain{}\rightmark\footnotesize\thepage}

\section{Full description of subsample-and-aggregate parameter settings}\label{app::lmm-settings}
Below are the full parameter settings used to create Table~\ref{tab:diff-k}. Each scenario was ran 50 times. 
\begin{itemize}
    \item Where required, the bound on the norm of the vector of estimates was set to $2\times 40\sqrt{d}=80\sqrt{6}$. 
    \item Where required, the lower bound on each coordinate was set to $\ell=-40\sqrt{6}$, and the upper bound was set to $u=40\sqrt{6}$.
    \item For the estimator of \citet{Yu2024}: $T$ was set to $\floor{\log n}$ and other parameters were left at the default values, as per the code provided kindly from \citet{Yu2024}. 
    \item For $\tilde\mu_p'$: $C=1$, $\eta=0$ and $\beta=1.001$.
    \item For the estimator of \citet{Biswas2020}: the upper bound on the covariance was set to $50\sqrt{d}=50\sqrt{6}$ and the starting center was set to the 0 vector. Also, $R=80\sqrt{6}$ and $\beta=0.1$. Three iterations were used in both the mean and covariance steps, with $\rho$ set to $1/3c,1/2c,1/c$ with $c=(1/2+1/3+1)/2$. 
    \item For the estimator of \citet{2021Liub}: $\alpha=0.125$, $R=80\sqrt{6}$, and $\delta=0.01$ and other parameters were left at the default values, as per the code provided kindly from \citet{2021Liub}.
\end{itemize}
\section{Proof of Theorem \ref{thm::main-result}}\label{sec::proof-main}
We now prove Theorem \ref{thm::main-result}. Let $\tilde X=X-\mu$, $\xi_{q,X}$ be the (left continuous) $q$th quantile of the distribution of the random variable $X\sim F$, $\hat\xi_{q,X}$ be the (non-private) estimated $q$th quantile based on $X$, and $\tilde\xi_{q,X}$ the private estimate $q$th quantile based on a random sample of size $n$, whose observations have the same distribution as $X$. 
\begin{proof}[Proof of Theorem \ref{thm::main-result}]
We prove the result for pure differential privacy, where the result for $\tilde\mu'$ follows by the same logic. 
Consider the following inequalities: 
\begin{align}
\label{eqn::step_1}
      &\xi_{1-7\zeta/4}\leq \hat \xi_{1-3\zeta/2}\leq \xi_{1-5\zeta/4}\leq  \xi_{1-3\zeta/4}\leq  \hat \xi_{1-\zeta/2}\leq   \xi_{1-\zeta/4}\\    
\label{eqn::step_2}
    &\xi_{\zeta/4}\leq \hat \xi_{\zeta/2}\leq \xi_{3\zeta/4}\leq \xi_{5\zeta/4} \leq \hat \xi_{3\zeta/2}\leq   \xi_{7\zeta/4}.
\end{align}
Let the event $E$ be the event on which the above inequalities hold. 
The first step of the proof is to show that on $E$, we have that 
\begin{align}\label{eqn::ineq-1}
    &\hat \xi_{\zeta/2}\leq \tilde \xi_{\zeta,Y'}\leq \hat\xi_{3\zeta/2} \andd \hat \xi_{1-3\zeta/2}\leq \tilde \xi_{1-\zeta,Y'}\leq \hat\xi_{1-\zeta/2},
\end{align}
holds with high probability. 
Then, we will show that $E$ occurs with high probability.

The inequality \eqref{eqn::ineq-1} is implied by 
$\zeta/2\leq F_n(\tilde \xi_{\zeta,Y'})\leq \frac{3}{2}\zeta $ and $1-\frac{3}{2}\zeta \leq F_n(\tilde \xi_{1-\zeta,Y'})\leq 1-\zeta/2$, where one may recall that $F_n(x)=\sum_{i=1}^n \ind{Y_i'\leq x}/n$ for any $x\in \re$. Note that these inequalities are implied by the statements ``$F_{n}(\tilde\xi_{\zeta})$ is within $\zeta/2$ of $\zeta$'', and that ``$F_{n}(\tilde\xi_{1-\zeta})$ is within $\zeta/2$ of $1-\zeta$'', respectively. That is, taking $q=\zeta$ and $ t = \zeta/2$, we are concerned with the quantities $|q - F_n(\tilde\xi_{q})| < t$ and $|1 - q - F_n(\tilde\xi_{1-q})| < t$. 
To prove that these inequalities hold with high probability, we can apply Lemma \ref{lem::pq-bound}, once taking $q=\zeta$ and $ t = \zeta/2$ and again with $q=1-\zeta$ and $ t = \zeta/2$.

In order to apply Lemma \ref{lem::pq-bound}, we must have then that three inequalities hold. 
Namely, $\zeta<1$, and 
\begin{align}
\label{eqn::c1_E1}
    \beta&\leq \frac{u-\hat\xi_{\zeta/2}+1}{u-\hat\xi_{3\zeta/2}+1}, \text{and}\\
\label{eqn::c2_E1}
    \beta&\leq \frac{\hat\xi_{1-\zeta/2}-\ell+1}{\hat\xi_{1-3\zeta/2}-\ell+1}.
\end{align}
First, $\zeta<1$ holds by Condition \ref{cond::bounds}. 
For \eqref{eqn::c1_E1}, note that \eqref{eqn::step_2} directly implies
\begin{align*}
   \frac{u-\hat\xi_{\zeta/2}+1}{u-\hat\xi_{3\zeta/2}+1}
    &=1+\frac{\hat\xi_{3\zeta/2}-\hat\xi_{\zeta/2}}{u-\hat\xi_{3\zeta/2}+1}\\
    &>1+\frac{\hat\xi_{3\zeta/2}-\hat\xi_{\zeta/2}}{u-\xi_{5\zeta/4}+1}\\
    &>1+\frac{\xi_{5\zeta/4}-\xi_{3\zeta/4}}{u-\xi_{5\zeta/4}+1}>1+b_n. 
\end{align*}
Similarly, for \eqref{eqn::c2_E1}, \eqref{eqn::step_1} directly implies that
\begin{align*}
   \frac{\hat\xi_{1-\zeta/2}-\ell+1}{\hat\xi_{1-3\zeta/2}-\ell+1} &=1+\frac{\hat\xi_{1-\zeta/2}-\hat\xi_{1-3\zeta/2}}{\hat\xi_{1-3\zeta/2}-\ell+1}>1+\frac{\xi_{1-3\zeta/4}-\xi_{1-5\zeta/4}}{\xi_{1-5\zeta/4}-\ell+1}>1+b_n. 
\end{align*}
Condition \ref{cond::beta} directly says that $\beta<1+b_n$, giving \eqref{eqn::c1_E1} and \eqref{eqn::c2_E1}. Thus, we can apply Lemma \ref{lem::pq-bound}. 
Applying Lemma~\ref{lem::pq-bound} gives 
\begin{align*}
    \Prr{\left|\zeta - F_n(\tilde\xi_{\zeta})\right| > \zeta/2} &\leq \left(\frac{\log(u - \hat\xi_{3\zeta/2} + 1)}{\log\beta} + 1\right)\exp(-n\zeta(\epsilon_1 \wedge \epsilon_2)/2),
\end{align*}
and
\begin{align*}
    \Prr{\left|1-\frac{\zeta}{2} - F_n(\tilde\xi_{1-\zeta/2})\right| > \frac{\zeta}{2}} &\leq \left(\frac{\log(\hat\xi_{1-3\zeta/2} - \ell + 1)}{\log\beta} + 1\right)\\
    &\hspace{5em}\times\exp\left(-\frac{n\zeta(\epsilon_1 \wedge \epsilon_2)}{2}\right).
\end{align*}

As a result, on $E$, \eqref{eqn::ineq-1} fails to hold with probability at most, \begin{align*}&\left(\frac{\log(u - \hat\xi_{3\zeta/2} + 1)}{\log\beta} + 1\right)\exp(-n\zeta(\epsilon_1 \wedge \epsilon_2)/2)\\
&\hspace{10em}  + \left(\frac{\log(\hat\xi_{1-3\zeta/2} - \ell + 1)}{\log\beta} + 1\right)\exp(-n\zeta(\epsilon_1 \wedge \epsilon_2)/2) \\ 
&= \left(\frac{\log(u - \hat\xi_{3\zeta/2} + 1) + \log(\hat\xi_{1-3\zeta/2} - \ell + 1)}{\log\beta} + 2\right)\exp(-n\zeta(\epsilon_1 \wedge \epsilon_2)/2)  \\
&\leq \left(\frac{\log(u - \xi_{5\zeta/4} + 1) + \log(\xi_{1-5\zeta/4} - \ell + 1)}{\log\beta} + 2\right)\exp(-n\zeta(\epsilon_1 \wedge \epsilon_2)/2)  \\
&\leq 2\left(\frac{\log(u - \ell + 1)}{\log\beta} + 1\right)\exp(-n\zeta(\epsilon_1 \wedge \epsilon_2)/2).
\end{align*}
The final inequality follows applying \eqref{eqn::step_1} and \eqref{eqn::step_2} and Condition~\ref{cond::bounds}.
Consider that, since $\varepsilon_1\wedge\varepsilon_2 \geq 3/56$ and $$\zeta \leq \frac{112}{3n}\log\left(\frac{32}{\delta}\left[\frac{\beta(u-\ell)}{\beta-1}\vee 1\right]\right).$$ 
Combining these results yields \begin{align*}
    &2\left(\frac{\log((u-\ell)+1)}{\log\beta} + 1\right)\exp\left(-\frac{n\zeta(\varepsilon_1\wedge\varepsilon_2)}{2}\right) \\
    &\indent\leq 2\left(\frac{\log((u-\ell)+1)}{\log\beta} + 1\right)\exp\left(-\frac{n\frac{112}{3n}\log\left(\frac{32}{\delta}\left[\frac{\beta(u-\ell)}{\beta-1}\vee 1\right]\right)(3/56)}{2}\right) \\
    &\indent\leq 2\left(\frac{(u-\ell)\beta}{\beta-1} +  1\right)\frac{\delta}{32}\left(\frac{\beta(u-\ell)}{\beta-1}\vee 1\right)^{-1} \\
    &\indent\leq \frac{\delta}{8}\left(\frac{\beta(u-\ell)}{\beta-1}\vee 1\right)\left(\frac{\beta(u-\ell)}{\beta-1}\vee 1\right)^{-1} \\
    &\indent= \frac{\delta}{8}.
\end{align*}

Thus, combining the preceding statements yields that, on $E$, we have that \eqref{eqn::ineq-1} holds with probability at least 
$$1-2\left(\frac{\log((u-\ell)+1)}{\log\beta} + 1\right)\exp\left(-\frac{n\zeta(\varepsilon_1\wedge\varepsilon_2)}{2}\right) \geq 1 - \frac{\delta}{8}.$$

The next step is to show that $E$ occurs, i.e., the inequalities \eqref{eqn::step_1} and \eqref{eqn::step_2} hold, with high probability. 
Call the event on which the inequalities \eqref{eqn::ineq-1} hold $A_1$. 
The inequalities \eqref{eqn::step_1} and \eqref{eqn::step_2} when combined with \eqref{eqn::ineq-1} imply,  
\begin{align*}
    &\xi_{\zeta/4,}< \tilde \xi_{\zeta,Y'}< \xi_{7\zeta/4} < \xi_{4\zeta} \andd \xi_{1-4\zeta} < \xi_{1-7\zeta/4}< \tilde \xi_{1-\zeta,Y'}<   \xi_{1-\zeta/4}.
\end{align*}
The next portion of the proof is similar to that of Theorem 1 of \citet{Lugosi2021}, however some steps are different, as that proof contains an error, see Appendix~\ref{app::lm-error}. 
Namely, one step uses the inequalities $$\E\phi_{\xi_{2\zeta,\tilde X},\xi_{1-\zeta/2\tilde X}}(X)\leq \E(\tilde X\ind{X\geq \xi_{1-\zeta/2,X}}) ,$$
and
$$\E\phi_{\xi_{2\zeta,\tilde X},\xi_{1-\zeta/2,X}}(X)\geq -\E(\tilde X\ind{X\leq \xi_{2\zeta,X}}),$$
which do not hold, see page 400 and 401 of \citet{Lugosi2021}. 
To this end, observe that $\Prr{\tilde X\geq \xi_{1-7\zeta/4,\tilde X}}=7\zeta/4$.
Recall that $\tilde{X} = X - \mu$.
Define $W_{k-} = \sum_{i=1}^n \ind{Y_i \leq \mu + \xi_{k, \tilde X}}$ and $W_{k+} = \sum_{i=1}^n \ind{Y_i \geq \mu + \xi_{k, \tilde X}}$. 
Then, note that \begin{align*}
    \Prr{W_{k\zeta-} \leq an\zeta} &= \Prr{W_{k\zeta-} - nk\zeta \leq n\zeta(a - k)} \\
    &= \Prr{nk - W_{k\zeta-} \geq n\zeta(k-a)} \\
    &\leq \exp\left(-\frac{n^2\zeta^2(k-a)^2/2}{nk\zeta(1-k\zeta) + n\zeta(k-a)/3}\right)\\
    &= \exp\left(-\frac{n\zeta(k-a)^2/2}{k(1-k\zeta) + (k-a)/3}\right).
\end{align*} This holds so long as $k \geq a$, by Bernstein's inequality. The same argument holds for $P(W_{(1-k\zeta)+} \leq an\zeta)$, giving the same upper bound. 
Effectively the same argument gives 
{\small
\begin{align*}
    \Prr{W_{(1-k\zeta)-} \leq n(1-a'\zeta)} &= \Prr{W_{(1-k\zeta)-} - n(1-k\zeta) \leq n(1-a'\zeta - (1-k\zeta))} \\
    &= \Prr{n(1-k\zeta) - W_{(1-k\zeta)-} \geq n\zeta(a'-k)} \\
    &\leq \exp\left(-\frac{n^2\zeta^2(a'-k)^2/2}{nk\zeta(1-k\zeta) + n\zeta(a'-k)/3}\right) \\
    &= \exp\left(-\frac{n\zeta(a'-k)^2/2}{k(1-k\zeta) + (a'-k)/3}\right).
\end{align*}}
This holds so long as $k \leq a'$. Symmetric arguments give the same bound for $P(W_{k\zeta+} \leq n(1-a'\zeta))$.

Suppose $W_1=W_{(1-4\zeta)+}$, $W_2=W_{(1-\zeta/4)-}$, $W_3=W_{4\zeta-}$, and $W_4=W_{\zeta/4+}$. 
Then, using the above results, 
{\small
\begin{align*}
P(W_1 \leq 17n\zeta/8) &\leq  \exp\left(-\frac{n\zeta(4-17/8)^2/2}{4(1-4\zeta) + (4-17/8)/3}\right) = \exp\left(-\frac{225n\zeta}{592 - 2048\zeta}\right) \\ 
P(W_2 \leq n(1-3\zeta/8)) &\leq  \exp\left(-\frac{n\zeta(3/8-1/4)^2/2}{1/4(1-1/4\zeta) + (3/8-1/4)/3}\right) = \exp\left(-\frac{3n\zeta}{112 - 24\zeta}\right) \\ 
P(W_3 \leq 17n\zeta/8) &\leq \exp\left(-\frac{n\zeta(4-17/8)^2/2}{4(1-4\zeta) + (4-17/8)/3}\right) = \exp\left(-\frac{225n\zeta}{592 - 2048\zeta}\right) \\ 
P(W_4 \leq n(1-3\zeta/8)) &\leq   \exp\left(-\frac{n\zeta(3/8-1/4)^2/2}{1/4(1-1/4\zeta) + (3/8-1/4)/3}\right) = \exp\left(-\frac{3n\zeta}{112 - 24\zeta}\right).
\end{align*}} 
Note that, for $0\leq \zeta \leq 1/4$, 
$$\exp\left(-\frac{225n\zeta}{592 - 2048\zeta}\right) \leq \exp\left(-\frac{3n\zeta}{112 - 24\zeta}\right).$$ 
Thus, with probability at least $1 - 4\exp\left(-\frac{3n\zeta}{112 - 24\zeta}\right) \geq 1 - 4\exp\left(-\frac{3n\zeta}{112}\right)$, we have that 
{\small\[A_2 = \{W_1 \geq 17n\zeta/8\} \cap \{W_2 \geq n(1-3\zeta/8)\} \cap \{W_3 \geq 17n\zeta/8\} \cap \{W_4 \geq n(1-3\zeta/8)\}\]}
holds. 

Recall that $$\zeta=16\eta+\frac{112}{3}\frac{\log(32(\beta(u-\ell) /(\beta-1)\vee 1)/\delta)}{n} \geq \frac{112}{3}\frac{\log(32/\delta)}{n}.$$
Then, plugging this bound into the previously derived probability, we find that $A_2$ holds with probability at least $1 - \delta/8$. 

Now, using the fact that $\eta\leq \zeta/16$, on $A_2$, we have that the following inequalities hold:
\begin{align*}
    &W_1\geq \frac{17}{8}\zeta n\\
    &\hspace{2em}\implies \sum_{i=1}^n \ind{Y_i'\geq \mu+ \xi_{1-4\zeta,\tilde X}}\geq \frac{17}{8}\zeta n-2\eta n\geq \frac{17}{8}\zeta n-n\zeta /8\geq 2n\zeta,\\
    &W_2\geq (1-3\zeta/8)n\\
    &\hspace{2em}\implies \sum_{i=1}^n \ind{Y_i'\geq \mu+ \xi_{1-\zeta/4,\tilde X}}\geq (1-3\zeta/8)n-2\eta n\geq n(1-\zeta/2).
\end{align*}
Similar arguments give $W_3 \geq 2n\zeta$ and $W_4 \geq n(1-\zeta/2)$, on $A_2$. 

These inequalities state that, on $A_2$, the number of centered points in the corrupted sample above $\xi_{1-4\zeta,\tilde X}$ is at least $2n\zeta$. 
In addition, on $A_2$, observe that the number of centered points in the corrupted sample below $\xi_{1-\zeta/4,\tilde X}$ is at least $n(1-\zeta/2)$. 
As a result, \eqref{eqn::step_1} holds on $A_2$. 
An analogous argument yields \eqref{eqn::step_2} on $A_2$. 
Therefore, $E\subset A_2$. 

Letting $A_3=A_1\cap A_2$, using the previous inequality, the definition of $\zeta$ we have that $A_3$ occurs with probability at least
\begin{equation}
    \begin{split}
        1- \frac{\delta}{8} - \frac{\delta}{8} = 1- \frac{\delta}{4}. \label{eqn::prob_1}
    \end{split}
\end{equation}
We now show that, on $A_3$, the uncorrupted mean estimate is close to the true mean with high probability. 
Observe that on $A_3$, 
\begin{align*}
    \frac{1}{n}\sum_{i=1}^n\phi_{\tilde\xi_{\zeta},\tilde\xi_{1-\zeta}}(X_i)&\leq \frac{1}{n}\sum_{i=1}^n\phi_{\xi_{4\zeta,X},\xi_{1-\zeta/4,X}}(X_i)\\
    &\leq \E\phi_{\xi_{4\zeta,X},\xi_{1-\zeta/4,X}}(X_1)\\
    &\hspace{2em}+\frac{1}{n}\sum_{i=1}^n[\phi_{\xi_{4\zeta,X},\xi_{1-\zeta/4,X}}(X_i)-\E\phi_{\xi_{4\zeta,X},\xi_{1-\zeta/4,X}}(X_1)].
\end{align*}

Observe the following equality
{\small
\begin{align*}
    \E\phi_{\xi_{4\zeta, \tilde X},\xi_{1-\zeta/4,\tilde X}}(\tilde X_1)&=\E(\xi_{4\zeta, \tilde X}\ind{\tilde X\leq \xi_{4\zeta, \tilde X}})+\E(\xi_{1-\zeta/4,\tilde X}\ind{\tilde X\geq \xi_{1-\zeta/4,\tilde X}})\\
    &\hspace{9em}+\E(\tilde X \ind{\xi_{4\zeta, \tilde X}< \tilde X< \xi_{1-\zeta/4,\tilde X}})\\
    &=\E(\xi_{4\zeta, \tilde X}\ind{\tilde X\leq \xi_{4\zeta, \tilde X}})+\E(\xi_{1-\zeta/4,\tilde X}\ind{\tilde X\geq \xi_{1-\zeta/4,\tilde X}})\\
    &\hspace{9em}-\E(\tilde X (\ind{\tilde X\geq \xi_{1-\zeta/4,\tilde X}}+\ind{\tilde X\leq \xi_{4\zeta, \tilde X}}))\\
    &=\E((\xi_{4\zeta, \tilde X}-\tilde X)\ind{\tilde X\leq \xi_{4\zeta, \tilde X}})\\
    &\hspace{9em}+\E((\xi_{1-\zeta/4,\tilde X}-\tilde X)\ind{\tilde X\geq \xi_{1-\zeta/4,\tilde X}})\\
    &=\E(\xi_{4\zeta, \tilde X}\vee \tilde X)+\E(\xi_{1-\zeta/4,\tilde X}\wedge\tilde X)\\
    &=|\E(\xi_{4\zeta, \tilde X}\vee \tilde X)|-|\E(\xi_{1-\zeta/4,\tilde X}\wedge\tilde X)|.
\end{align*}}
Now, we have that 
\begin{align*}
    \E\phi_{\xi_{4\zeta,X},\xi_{1-\zeta/4,X}}(X_1)&=\mu+\E\phi_{\xi_{4\zeta, \tilde X},\xi_{1-\zeta/4,\tilde X}}(\tilde X_1)\\
    &= \mu+|\E(\xi_{4\zeta, \tilde X}\vee \tilde X)|-|\E(\xi_{1-\zeta/4,\tilde X}\wedge\tilde X)|\\
    &\leq \mu+\mathcal{E}(4\zeta,X),
\end{align*}
where $\mathcal{E}(a,X)= |\E(\xi_{a, \tilde X}\vee \tilde X)|\vee |\E(\xi_{1-a,\tilde X}\wedge\tilde X)|$. 
Similarly,
\begin{align*}
    \E\phi_{\xi_{4\zeta,X},\xi_{1-\zeta/4,X}}(X_1)&= \mu-|\E(\xi_{1-\zeta/4,\tilde X}\wedge\tilde X)|+|\E(\xi_{4\zeta, \tilde X}\vee \tilde X)|\\
    &\geq  \mu-\mathcal{E}(\zeta/4,X). 
\end{align*}

Now, let 
$$\overline{\omega}=\frac{1}{n}\sum_{i=1}^n \omega_i=\frac{1}{n}\sum_{i=1}^n[\phi_{\xi_{4\zeta, X},\xi_{1-\zeta/4, X}}(X_i)-\E\phi_{\xi_{4\zeta, X},\xi_{1-\zeta/4,X}}(X_1)]$$
and observe that $\overline{\omega}$ is an independent sum of centered random variables, bounded above by $(\xi_{1-\zeta/4,\tilde X}+\mathcal{E}(\zeta/4,X))/n$. 
Moreover, $\text{var}(\omega_i) \leq \sigma^2$.
We then have that by Bernstein's inequality
\begin{align*}
    \Prr{\overline{\omega}\geq t}&\leq \exp\left(-\frac{t^2/2}{\var(\omega_1) /n+(\xi_{1-\zeta/4,\tilde X}+\mathcal{E}(\zeta/4,X))t/3n}\right)\\
    &= \exp\left(-\frac{nt^2/2}{\var(\omega_1) +(\xi_{1-\zeta/4,\tilde X}+\mathcal{E}(\zeta/4,X))t/3}\right)\\
    &\leq \exp\left(-\frac{nt^2/2}{\sigma^2 +(\xi_{1-\zeta/4,\tilde X}+\mathcal{E}(\zeta/4,X))t/3}\right).
\end{align*}
Thus, taking 
{\small{\begin{align*}
    t &= \frac{\sqrt{[2\left(\mathcal{E}(\zeta/4, X) + \xi_{1-\zeta/4,\tilde X}\right)\log(1/p)]^2 + 72n\sigma^2\log(1/p)} + 2\left(\mathcal{E}(\zeta/4, X) + \xi_{1-\zeta/4,\tilde X}\right)\log(1/p)}{6n} \\ 
    &= \frac{\sqrt{[\left(\mathcal{E}(\zeta/4, X) + \xi_{1-\zeta/4,\tilde X}\right)\log(1/p)]^2 + 18n\sigma^2\log(1/p)} + \left(\mathcal{E}(\zeta/4, X) + \xi_{1-\zeta/4,\tilde X}\right)\log(1/p)}{3n}
\end{align*}}}
renders $P(\overline{\omega} \geq t) \leq p$. 
Consider that, with $$t^* = \sigma\sqrt{\frac{6\log(1/p)}{n}} + \frac{2\left(\mathcal{E}(\zeta/4, X) + \xi_{1-\zeta/4,\tilde X}\right)\log(1/p)}{3n},$$ we have $t \leq t^*$ and thus $P(\overline{\omega} \geq t^*) \leq P(\overline{\omega} \geq t) \leq p$.
Taking $p=\delta/4$, combined with \eqref{eqn::prob_1} and the assumption that $\delta\geq 4e^{-n}$, yields that, with probability at least $1 - \delta/4 - \delta/4 = 1 - \delta/2$,  there exists universal constants $c,c'>0$ such that
\begin{align*}
  \frac{1}{n}\sum_{i=1}^n\phi_{\xi_{\zeta/4, X},\xi_{1-\zeta/4, X}}(X_i)&\leq \mu+\mathcal{E}(4\zeta,X)\\
    &\hspace{2em}+c\Bigg( \sigma\sqrt{ \frac{\log(4/\delta)}{n}}+\frac{2\xi_{1-\zeta/4,\tilde X}\log(4/\delta)}{3n}\\
    &\hspace{8em}+\frac{2\mathcal{E}(\zeta/4,X)\log(4/\delta)}{3n}\Bigg)\\
   &\leq \mu+c'\mathcal{E}(\zeta/4,X)\\
    &\hspace{4em}+ c\left( \sigma\sqrt{ \frac{\log(4/\delta)}{n}}+\frac{2\xi_{1-\zeta/4,\tilde X}\log(4/\delta)}{3n}\right).
\end{align*}

Applying Chebyshev's inequality yields that 
$$\frac{\zeta}{4}=\Prr{\tilde X\geq \xi_{1-\zeta/4,\tilde X}}\leq \frac{\sigma^2}{\xi_{1-\zeta/4,\tilde X}^2}\implies \xi_{1-\zeta/4,\tilde X}\leq\frac{\sqrt{4}\sigma}{\sqrt{\zeta}}.$$
Since $\delta/4 \geq e^{-n}$ then $\log(4/\delta) \leq n$, and so $\sqrt{\log(4/\delta)/n} \geq \log(4/\delta)/n$. 
Then, using the result from Chebyshev's inequality, with probability at least $1-\delta/2$,
\begin{align*}
   \frac{1}{n}\sum_{i=1}^n\phi_{\xi_{4\zeta, X},\xi_{1-\zeta/4, X}}(X_i)- \mu&\lesssim \mathcal{E}(\zeta/4,X)+ \sigma\sqrt{ \frac{\log(4/\delta)}{n}}.
\end{align*}
Lastly, we bound $\mathcal{E}(\zeta/4,X)$ above. 
Using the triangle inequality, Holder's inequality and Chebyshev's inequality,
{\small
\begin{align*}
    |E( (X-\xi_{1-a}) \ind{ X>\xi_{1-a})}|&\leq \E(|\tilde X\ind{X> \xi_{1-a}}|)+\E(|\xi_{1-a}| \ind{ X>\xi_{1-a}})\\
    &\leq\sigma\Prr{X> \xi_{1-a}}^{1/2}+\sigma\sqrt{a}\\
    &\lesssim \sigma\sqrt{a}.
\end{align*}}
Similarly, $|E( (X-\xi_{a}) \ind{ X<\xi_{a})}|\lesssim \sigma\sqrt{a} $. 
Taken together, these two inequalities give $\mathcal{E}(\zeta/4,X) \lesssim \sigma\sqrt{\zeta}$. 
Thus, with probability at least $1-\delta/2$,
\begin{align*}
   \frac{1}{n}\sum_{i=1}^n\phi_{\xi_{4\zeta, X},\xi_{1-\zeta/4, X}}(X_i)- \mu&\lesssim \sigma\sqrt{\zeta}+ \sigma\sqrt{ \frac{\log(4/\delta)}{n}}.
\end{align*}
An analogous argument holds to give a similar lower bound on the same quantity. 

Next, we consider $$\left|\frac{1}{n}\sum_{i=1}^n\phi_{\xi_{4\zeta, X},\xi_{1-\zeta/4, X}}(X_i)-\frac{1}{n}\sum_{i=1}^n\phi_{\xi_{4\zeta, X},\xi_{1-\zeta/4, X}}(X_i')\right|,$$ following the argument in \citet{Lugosi2021}. 
Using Chebyshev's inequality, note that the maximal gap for a single term in the summation, $$\left|\phi_{\xi_{4\zeta, X},\xi_{1-\zeta/4, X}}(X_i) - \phi_{\xi_{4\zeta, X},\xi_{1-\zeta/4, X}}(X_i')\right| \leq |\xi_{1-\zeta/4, X}| + |\xi_{4\zeta, X}| \lesssim \sigma/\sqrt{\zeta}.$$ 
Since $\phi_{\xi_{4\zeta, X},\xi_{1-\zeta/4, X}}(X_i) = \phi_{\xi_{4\zeta, X},\xi_{1-\zeta/4, X}}(X_i')$ for all but $2\eta n$ (or fewer) points, we get that $$\left|\frac{1}{n}\sum_{i=1}^n\phi_{\xi_{4\zeta, X},\xi_{1-\zeta/4, X}}(X_i)-\frac{1}{n}\sum_{i=1}^n\phi_{\xi_{4\zeta, X},\xi_{1-\zeta/4, X}}(X_i')\right|\lesssim \eta\sigma/\sqrt{\zeta}\lesssim \sigma\sqrt{\zeta}.$$ The last inequality follows since, by assumption, $\eta \leq \zeta/16$.

Recall that $$\tilde\mu=\frac{1}{n}\sum_{i=1}^n\phi_{\tilde\xi_{\zeta},\tilde\xi_{1-\zeta}}(X_i')+Z_1\frac{(\tilde\xi_{1-\zeta}-\tilde\xi_{\zeta})}{n\varepsilon_3} \leq \frac{1}{n}\sum_{i=1}^n\phi_{\xi_{4\zeta},\xi_{1-\zeta/4}}(X_i')+Z_1\frac{(\tilde\xi_{1-\zeta}-\tilde\xi_{\zeta})}{n\varepsilon_3}.$$ 
Combining this with \eqref{eqn::prob_1}, the preceding tail argument, and the fact that $\Prr{Z_1\geq 1}= e^{-n}\leq \delta/4$, we have that with probability at least $1-\delta,$
\begin{multline*}
    |\tilde \mu-\mu|\lesssim \sigma\sqrt{\zeta}+ \sigma\sqrt{ \frac{\log(4/\delta)}{n}}+\frac{Z_1(\tilde\xi_{1-\zeta}-\tilde\xi_{\zeta})}{n\varepsilon_3}\\
    \lesssim \sigma\sqrt{\zeta}+ \sigma\sqrt{ \frac{\log(4/\delta)}{n}}+\frac{Z_1(\xi_{1-\zeta/4}-\xi_{\zeta/4})}{n\varepsilon_3}\\
   \lesssim \sigma\sqrt{\zeta}+ \sigma\sqrt{ \frac{\log(4/\delta)}{n}}+\frac{\xi_{1-\zeta/4}-\xi_{\zeta/4}}{n\varepsilon_3}.
\end{multline*}

It suffices to bound the term $\xi_{1-\zeta/4}-\xi_{\zeta/4}$. 
To this end, note that $\xi_{1-\zeta/4}-\xi_{\zeta/4}=\xi_{1-\zeta/4,X}-\xi_{\zeta/4,X}$. 
In addition, we have that $\zeta\gtrsim 1/n$ and so 
$$\xi_{1-\zeta/4}-\xi_{\zeta/4}\leq 2(\xi_{1-\zeta/4,X}\vee  (-\xi_{\zeta/4}))\leq 2(\xi_{1-1/4n,X}\vee  (-\xi_{1/4n,X})).$$ 
Suppose that $\xi_{1-1/4n,X}>0$, then Chebyshev's inequality implies that
\begin{align*}
    1/4n&=\Prr{X\geq \xi_{1-1/4n,X}}\leq \Prr{|X|\geq \xi_{1-1/4n,X}}\leq \frac{\sigma^2}{\xi_{1-1/4n,X}^2}\\
    &\hspace{6em}\implies  \xi_{1-1/4n,X}\lesssim \sigma\sqrt{n}.
\end{align*}
In addition, if $\xi_{1-1/4n,X}<0$, then $\xi_{1-1/4n,X}\lesssim \sigma\sqrt{n}$. 
Similarly, suppose that $\xi_{1/4n,X}<0$. Again, Chebyshev's inequality implies that
\begin{align*}
    1/4n&=\Prr{X\leq \xi_{1/4n,X}}\leq \Prr{|X|\geq -\xi_{1/4n,X}}\leq \frac{\sigma^2}{\xi_{1/4n,X}^2}\\
    &\hspace{6em}\implies  -\xi_{1/4n,X}\lesssim \sigma\sqrt{n}.
\end{align*}
Thus, $\xi_{1-\zeta/4}-\xi_{\zeta/4}\lesssim\sigma\sqrt{n}.$ 
\end{proof}
\section{Zero-concentrated DP quantiles (zCDP)}
We now define a version of the estimator that satisfies a weaker notion of privacy called zero-concentrated DP. 
The algorithm proceeds as follows: 
Let $V,V_1,V_2,\ldots\sim \cN(0,1)$. 
Given $\beta>1$, a lower bound $\ell$, data set $\{Y_i\}_{i=1}^n$ with empirical distribution $F_n$, and quantile $q>1/2$, we do the following:
\begin{enumerate}
    \item First, compute $\hat q=q+V/n\sqrt{\rho}$. 
    \item While $F_{n}(\beta^i+\ell -1)+V_i/n\sqrt{\rho}< \hat q$, increase $i$ by one at a time. Otherwise, stop. 
    \item If we stop at $k$, then output $\beta^k+\ell -1$.
\end{enumerate}
If instead $1/2>q\to 0$, we negate everything and take an upper bound as input instead of lower bound. The idea is that the upper bound is far from the minimum, etc. 
Call these estimates $\tilde\xi_{q,\sam{Y}{n}}$. We will omit $\sam{Y}{n}$ when it is obvious.

We first show that adding Gaussian noise with variance $1/n^2\rho$ to the queries in the private quantile estimation results in $\rho$-zCDP. 
\begin{lemma}\label{lem::PQE-zCDP}
For all $0\leq q\leq 1$ and $\rho_1,\rho_2>0$, $\tilde\xi_q'$ satisfies $(\rho_1+\rho_2)$-zCDP.
\end{lemma} 
\begin{proof}
Consider neighboring datasets $\mathbf{x}_n,\mathbf{x}_n'$ with empirical distribution functions, $F_n$ and $F_n'$, respectively. 
Let $\D_\alpha(P||Q)$ denotes the $\alpha$-R\'enyi Divergence between two measures $P$ and $Q$. 
That is, $$\D_\alpha(P||Q) = \frac{1}{\alpha - 1}\log\int\left(\frac{dP}{dQ}\right)^\alpha dQ.$$%
For two random variables on the same probability space, $X\sim P$ and $Y\sim Q$, let $\D_\alpha(X||Y)=\D_\alpha(P||Q)$. 
The aim is to show that for any $\alpha>1$, $1/2<q<1$ and any pair of adjacent databases $\mathbf{x}_n,\mathbf{x}_n'$, we have that
$$\D_\alpha(\tilde\xi_{q,\mathbf{x}_n}||\tilde\xi_{q,\mathbf{x}_n'})\leq \alpha\rho.$$ 
To this end, let $K$ be the random integer at which the process stops on the original dataset $\mathbf{x}_n$, and let $K'$ be the random integer at which the process stops on the neighboring dataset $\mathbf{x}_n'$. 
It follows by post-processing, (see Lemma 15 of \cite{Bun2016}) that $\D_\alpha(\tilde\xi_{q,\mathbf{x}_n}||\tilde\xi_{q,\mathbf{x}_n'})\leq \D_{\alpha}(K||K')$, since $\tilde\xi_{q,\mathbf{x}_n'} = \beta^{K}+\ell-1$. 
It suffices to show that
$\D_{\alpha}(K||K')\leq \rho\alpha$ for all $\alpha>1$. 

To this end, let $M\eqd K$, let $\nu_j$ be the joint distribution of the sequence of random variables $Z=(K, V_1/n\sqrt{\rho},\ldots,V_{j-1}/n\sqrt{\rho}, V_{j+1}/n\sqrt{\rho},\ldots)$ and let $\nu_j'$ be the distribution of $Z'=(K', V_1/n\sqrt{\rho},\ldots,V_{j-1}/n\sqrt{\rho}, V_{j+1}/n\sqrt{\rho},\ldots)$. 
Take $I$ to be the function returning the first element, so that $I(Z)=K$ and $I(Z')=K'$. 
It follows by Lemma 15 of \cite{Bun2016} (post-processing) that 
\begin{equation*}
    \D_{\alpha}(K||K')=  \D_{\alpha}(I(\nu_j)||I(\nu_j'))\leq \D_{\alpha}(\nu_j||\nu_j').
\end{equation*}
We now consider $\D_{\alpha}(\nu_j||\nu_j')$ for arbitrary $j$. 
Let $Q_j$ and $Q_j'$ denote the distribution of $K$ and $K'$ conditional on $(V_1/n\sqrt{\rho},\ldots,V_{j-1}/n\sqrt{\rho}, V_{j+1}/n\sqrt{\rho},\ldots)$. 
Now, let $\mu_j$ be the joint distribution of the sequence of random variables $(V_1/n\sqrt{\rho},\ldots,V_{j-1}, V_{j+1}/n\sqrt{\rho},\ldots)$. 
Next, we have that by Lemma 15 of \cite{Bun2016} (composition), 
\begin{equation*}
\D_{\alpha}(\nu_j||\nu_j')\leq \sup_{j\geq 0}[\D_{\alpha}(Q_j||Q_j')+\D_{\alpha}(\mu_j||\mu_j)]=\sup_{j\geq 0}\D_{\alpha}(Q_j||Q_j'). 
\end{equation*}
It then suffices to show that, for any $j\geq 0$, we have that $\D_{\alpha}(Q_j||Q_j')\leq \alpha\rho$. 
To this end, let $(v_1/n\sqrt{\rho},\ldots,v_{j-1}/n\sqrt{\rho},v_{j+1}/n\sqrt{\rho},\ldots)$ be the realizations of the sequence of random variables $(V_1/n\sqrt{\rho},\ldots,V_{j-1}/n\sqrt{\rho}, V_{j+1}/n\sqrt{\rho},\ldots)$. 
Next, for any positive integer $k$, let $z_k=F_n(\beta^k+\ell -1)$, $\tau_k=\max_{i<k}[F_n(\beta^i+\ell -1)+v_i]$, $z_k'=F'_n(\beta^k+\ell -1)$ and $\tau_k'=\max_{i<k}[F_n'(\beta^i+\ell -1)+v_i]$. 
To prove the result, we first have that 
\begin{multline*}
   \Prr{K=k\Big| \{V_i=v_i\}_{i=1}^{k-1}, \{V_i=v_i\}_{i=k+1}^{\infty}}\\
   = \Prr{\tau_k\leq  \hat q <z_k+V_k/n\sqrt{\rho}\Big| \{V_i=v_i\}_{i=1}^{k-1}, \{V_i=v_i\}_{i=k+1}^{\infty}}\\ 
   = \Prr{\tau_k\leq  \hat q <z_k+V_k/n\sqrt{\rho}\Big| \{V_i=v_i\}_{i=1}^{k-1}}\coloneqq b_k.
\end{multline*}
Next, letting $\varphi(x)$ be the standard normal density, $q'(x)=q+x/n\sqrt{\rho} +\tau_k'-\tau_k$ and
$\kappa(x)=x+n\sqrt{\rho}(\tau_k'-\tau_k+z_k-z_k'),$
by definition, it holds that 
\begin{align*}
    b_k&=\int_{\re}\int_{\re}\ind{\tau_k< q+x/n\sqrt{\rho}<z_k+y/n\sqrt{\rho}}\varphi(x)\varphi(y)dxdy.
\end{align*}
In the above integral, it is helpful to note that $x$ corresponds to the random variable $V$, and $y$ corresponds to the random variable $V_k$. Moreover, the indicator depends on $(V_1,\dots,V_{k-1})$ only through $\tau_k$. Since $V$ and $V_k$ are both independent of $(V_1,\dots,V_{k-1})$, the conditional densities are equal to the marginal densities.
Continuing,
\begin{align*}
    b_k&=\int_{\re}\int_{\re}\ind{\tau_k< q+x/n\sqrt{\rho}<z_k+y/n\sqrt{\rho}}\varphi(x)\varphi(y)dxdy\\
    &=\int_{\re}\int_{\re}\ind{\tau_k<q'(x)-\tau_k'+\tau_k<\kappa(y)/n\sqrt{\rho}-\tau_k'+\tau_k+z_k'}\varphi(x)\varphi(y)dxdy\\
    &=\int_{\re}\int_{\re}\ind{\tau_k'<q'(x)<\kappa(y)/n\sqrt{\rho}+z_k'}\varphi(x)\varphi(y)dxdy.
\end{align*}
Next, we make the change of variable $u_1=\kappa(y)/n\sqrt{\rho}$ and $u_2\coloneqq u_2(x)=q'(x)-q$. 
\begin{align*}
   b_k&=(n^2\rho)^{-1}\int_{\re}\int_{\re}\ind{\tau_k'<u_2+q<u_1+z_k'-\tau_k'}\varphi(n\sqrt{\rho}(u_2+\tau_k-\tau_k'))\\
  &\hspace{10em} \times\varphi(n\sqrt{\rho}(u_1+\tau_k-z_k+z_k'-\tau_k'))du_1du_2\\
&=\Prr{\tau_k'<U_2+q<U_1+z_k'}\\
&=\Prr{\tau_k'<U_3+q<z_k'}\coloneqq h_k(U_3),
\end{align*}
where $U_1\sim \cN(\tau_k'-\tau_k-z_k'+z_k,(n^2\rho)^{-1}),U_2\sim \cN(\tau_k'-\tau_k,(n^2\rho)^{-1})$ and $U_3=U_2-U_1$, with $U_1\perp U_2$. 
By definition, i.e., the properties of the normal distribution, we have that $U_3\sim \cN(\tau_k'-\tau_k-\tau_k'+\tau_k+z_k'-z_k,2(n^2\rho)^{-1})=\cN(z_k'-z_k,2(n^2\rho)^{-1})$.
In addition, observe that by definition, 
\begin{align*}
    &\Prr{K'=k\Big| \{V_i=v_i\}_{i=1}^{k-1}, \{V_i=v_i\}_{i=k+1}^{\infty}}\\
    &\hspace{6em}=\Prr{\tau_k'\leq  \hat q <z_k'+V_k/n\sqrt{\rho}\Big| \{V_i=v_i\}_{i=1}^{k-1}, \{V_i=v_i\}_{i=k+1}^{\infty}}\\
     &\hspace{6em}=h_k((n\sqrt{\rho})^{-1}(V-V_j)).
\end{align*}
Using these facts yields that 
{\small\begin{align*}
\D_{\alpha}(Q_j||Q_j')&=\frac{1}{\alpha-1}\sup_{j>0}\E_{M}\left[\left(\frac{\Prr{K=M\Big| \{V_i=v_i\}_{i=1}^{j-1}, \{V_i=v_i\}_{i=j+1}^{\infty}}}{\Prr{K'=M\Big| \{V_i=v_i\}_{i=1}^{j-1}, \{V_i=v_i\}_{i=j+1}^{\infty}}}\right)^{\alpha-1}\right]\\
&=\frac{1}{\alpha-1}\sup_{j>0}\E_{M}\left[\left(\frac{h_j(U_3)}{h_j((n\sqrt{\rho})^{-1}(V-V_j))}\right)^{\alpha-1}\right]\\
&\leq \frac{1}{\alpha-1}\sup_{j>0}\E_{M}\left[\left(\frac{U_3}{(n\sqrt{\rho})^{-1}(V-V_j)}\right)^{\alpha-1}\right]\\
&=\alpha 2\rho n^2\E(U_3)^2/2\\
&=\alpha \rho n^2\E(U_3)^2,
\end{align*}}
where the second last line uses the fact that the R\'enyi divergence between two normal distributions with variance $\sigma^2$ and means $\mu_1,\mu_2$ is given by $\frac{1}{2}\alpha(\mu_1-\mu_2)^2/2\sigma^2.$
Now, we just need to bound the expectation $\E(U_3)$, from which it follows by definition that $-1/n<\E(U_3)<1/n$. 
As a result, we have that $\D_{\alpha}(Q_j||Q_j')\leq\alpha \rho ,$ which completes the proof. 
\end{proof}

\section{Proving finite sample bounds for private quantiles} 
The following lemma gives a concentration result related to the purely differentially private quantile computed from any sample. 
\begin{lemma}\label{lem::pq-bound}
Given the corrupted sample $Y_1',\ldots,Y_n'$, and $t>0$, we have that for all $1/2<q\leq 1$, $\ell\leq \hat\xi_{q}$, and $1<\beta\leq \left.(\hat\xi_{q+t} -\ell+1)\right/(\hat\xi_{q-t} -\ell+1)$ it holds that
\begin{align*}
\Prr{|q-F_n(\tilde\xi_{q})|> t}&\leq \left(\frac{\log(\hat \xi_{ q-t}+1-\ell)}{\log\beta} + 1\right)\exp\left(-nt(\varepsilon_1\wedge\varepsilon_2) \right),
\end{align*}
and for all $0<q\leq 1/2$, $u\geq \hat\xi_{q}$, and $1<\beta\leq \left.(u-\hat\xi_{q-t} +1)\right/(u-\hat\xi_{q+t} +1)$ it holds that
\begin{align*}
\Prr{|q-F_n(\tilde\xi_{q})|> t}&\leq \left(\frac{\log(u-\hat \xi_{ q+t}+1)}{\log\beta} + 1\right)\exp\left(-nt(\varepsilon_1+\varepsilon_2) \right).
\end{align*}
\end{lemma}
\begin{proof}[Proof of Lemma \ref{lem::pq-bound}]
We start by demonstrating the first inequality. 
Let $\hat q=q+V/n\varepsilon_1\geq q$. 
Observe that $$\Prr{|q-F_n(\tilde\xi_{q})|> t}=\Prr{q-F_n(\tilde\xi_{q})> t}+\Prr{F_n(\tilde\xi_{q})-q> t}\coloneqq I+II.$$
First, we must ensure that there is some $k\in\bbN$ such that $F_n(\beta^k-1+\ell)$ that is within $t$ of $q$. 
Let $k$ be the integer at which the process terminates. 
We have that  
\begin{align*}
    q-F_n(\tilde\xi_{q})\geq t
    &\iff    \beta^k-1+\ell \leq  \hat\xi_{ q-t} \iff   k\leq \frac{\log(\hat\xi_{q-t} -\ell+1)}{\log \beta}\coloneqq k_1\\
    F_n(\tilde\xi_{q})-q\geq t
    &\iff    \beta^k-1+\ell \geq  \hat\xi_{ q+t} \iff   k\geq \frac{\log(\hat\xi_{q+t} -\ell+1)}{\log \beta}\coloneqq k_2.
\end{align*}
Thus, as long as there is an integer between $k_1$ and $k_2$, there is a valid $k \in\mathbb{N}$. This is equivalent to saying $1 \leq k_2 - k_1$, giving
\begin{align*}
    1 &\leq \frac{\log(\hat\xi_{q+t} -\ell+1)}{\log \beta}- \frac{\log(\hat\xi_{q-t} -\ell+1)}{\log \beta} 
    \implies \beta \leq \frac{\hat\xi_{q+t} -\ell+1}{\hat\xi_{q-t} -\ell+1}.  
\end{align*}
This holds by assumption, and so a valid $k$ must exist.
We now consider $I$. 

Using the definition of $k$, the fact that $V$ and $\{V_i\}_{i=1}^\infty$ are standard exponential random variables and a union bound, it holds that
\begin{align*}
\Prr{k<k_1}&\leq \sum_{i=1}^{k_1}\Prr{V_{i}/n\varepsilon_2\geq \hat{q}-F_n(\beta^{i}+\ell-1)}\\
&\leq \sum_{i=1}^{k_1}\Prr{V_{i}/n\varepsilon_2\geq \hat{q}-F_n(\beta^{k_1}+\ell-1)}\\
&= \sum_{i=1}^{k_1}\Prr{V_{i}/n\varepsilon_2\geq q + V/n\epsilon_1 - F_n((\hat\xi_{q-t}-\ell+1)+\ell-1)}\\
&\leq k_1\Prr{V_{k_1}-V\geq tn(\varepsilon_1\wedge\varepsilon_2)}\\
 &=\frac{\log(\hat \xi_{ q-t}+1-\ell)}{2\log\beta}\exp\left(-nt(\varepsilon_1\wedge\varepsilon_2) \right).
\end{align*}

Now, we consider $II$. 
Observe that 
\begin{align*}
II=\Prr{k> k_2}&= \prod_{i=1}^{k_2}\Prr{V_{i}/n\varepsilon_2+F_n(\beta^{i}+\ell-1)\leq \hat q}\\
&\leq\Prr{V_{i}/n\varepsilon_2+F_n(\beta^{k_2}+\ell-1)\leq \hat q}\\
&= \Prr{V_{i}/n\varepsilon_2+F_n((\hat\xi_{q+t}-\ell+1)+\ell-1)\leq q+V/n\epsilon_1}\\
&\leq\Prr{tn(\varepsilon_2\wedge\varepsilon_2) \leq V-V_{i}}\\
&\leq \frac{1}{2}\exp\left(-nt(\varepsilon_1\wedge\varepsilon_2) \right).
\end{align*}

Then, for the second inequality, since $0 < q \leq \frac{1}{2}$, the algorithm is run with $\ell = -u$, quantile $1-q$, and each data point negated, defining $Y = -X$. Then $\tilde{\xi}_{q,X} = -\tilde{\xi}_{q',Y}$, $\hat{\xi}_{q,X} = -\hat{\xi}_{q',Y}$, and $F_{n, X}(k) = 1 - F_{n,Y}(-k)$. 

Note that $0 < q \leq \frac{1}{2} \iff \frac{1}{2} < q' \leq 1$, $u \geq \hat\xi_{q, X} \iff \ell \leq \hat\xi_{q',Y}$, and 
\begin{multline*}
    1 \leq \beta \leq \frac{u - \hat\xi_{q-t,X} + 1}{u - \hat\xi_{q+t, X}+1} \iff 1 \leq \beta \leq \frac{-\ell  + \hat\xi_{1-(q-t),Y} + 1}{-\ell + \hat\xi_{1-(q+t), Y}+1}\\
    \iff 1 \leq \beta \leq \frac{\hat\xi_{q'+t,Y} - \ell + 1}{\hat\xi_{q'-t, Y} - \ell +1}.
\end{multline*}
Thus, the first inequality can be applied with $\ell = -u$, $q'$, and $Y$. This gives \[P(|q' - F_{n,Y}(\tilde{\xi}_{q'-t,Y})| > t) \leq \left(\frac{\log(\hat{\xi}_{q', Y} + 1 - \ell)}{\log(\beta)} + 1\right)\exp\left(-nt(\epsilon_1 \wedge \epsilon_2)\right).\] Note that \[|q' - F_{n,Y}(\tilde{\xi}_{q',Y})| = |1 - q - (1 - F_{n,X}(-\tilde{\xi}_{q',Y}))| = |F_{n,X}(\tilde{\xi}_{q,X}) - q|.\] For the upper bound, consider that \[\frac{\log(\hat{\xi}_{q', Y} + 1 - \ell)}{\log(\beta)} = \frac{\log(-\hat{\xi}_{1-(1-q-t), X} + 1 + u)}{\log(\beta)} = \frac{\log(u-\hat{\xi}_{q+t, X} + 1)}{\log(\beta)}. \]
This completes the proof. 
\end{proof}
The following lemma gives a concentration result related to the zero-concentrated differentially private quantile computed from any sample. 
\begin{lemma}\label{lem::pq-bound-zCDP}
Given the corrupted sample $Y_1',\ldots,Y_n'$, and $t>0$, we have that for all $1/2<q\leq 1$, $\ell\leq \hat\xi_{q}$, and $1<\beta\leq (\hat\xi_{q+t} -\ell+1)/(\hat\xi_{q-t} -\ell+1)$, we have that 
\begin{align*}
\Prr{|q-F_n(\tilde\xi_{q})|> t}&\leq \left(\frac{\log(\hat \xi_{ q-t}+1-\ell)}{\log\beta} + 1\right)\\
&\hspace{2em}\times \exp\left(-n(t\vee t^2)(\sqrt{\rho_1\vee \rho_1^2}\wedge\sqrt{\rho_2\vee \rho_2^2})/4 \right),
\end{align*}
and for all $0<q\leq 1/2$, $u\geq \hat\xi_{q}$, and $1<\beta\leq (u-\hat\xi_{q-t} +1)/(u-\hat\xi_{q+t} +1)$ it holds that
\begin{align*}
\Prr{|q-F_n(\tilde\xi_{q})|> t}&\leq \left(\frac{\log(u-\hat \xi_{ q+t}+1)}{\log\beta} + 1\right)\\
&\hspace{2em}\times \exp\left(-n(t\vee t^2)(\sqrt{\rho_1\vee \rho_1^2}\wedge\sqrt{\rho_2\vee \rho_2^2})/4 \right).
\end{align*}

\end{lemma}
\begin{proof}[Proof of Lemma~\ref{lem::pq-bound-zCDP}]
The proof mirrors that of Lemma~\ref{lem::pq-bound}. 
We start by demonstrating the first inequality. 
Let $\hat q'=q+V/n\sqrt{\rho}_1$. 
Observe that $$\Prr{|q-F_n(\tilde\xi_{q})|> t}=\Prr{q-F_n(\tilde\xi_{q})> t}+\Prr{F_n(\tilde\xi_{q})-q> t}\coloneqq I+II.$$
First, we must ensure that there is some $k\in\bbN$ such that $F_n(\beta^k-1+\ell)$ that is within $t$ of $q$. 
Let $k$ be the integer at which the process terminates. 
We have that  
\begin{align*}
    q-F_n(\tilde\xi_{q})\geq t
    &\iff    \beta^k-1+\ell \leq  \hat\xi_{ q-t} \iff   k\leq \frac{\log(\hat\xi_{q-t} -\ell+1)}{\log \beta}\coloneqq k_1\\
    F_n(\tilde\xi_{q})-q\geq t
    &\iff    \beta^k-1+\ell \geq  \hat\xi_{ q+t} \iff   k\geq \frac{\log(\hat\xi_{q+t} -\ell+1)}{\log \beta}\coloneqq k_2.
\end{align*}
Thus, as long as there is an integer between $k_1$ and $k_2$, there is a valid $k \in\mathbb{N}$. This is equivalent to saying $1 \leq k_2 - k_1$, giving
\begin{align*}
    1 &\leq \frac{\log(\hat\xi_{q+t} -\ell+1)}{\log \beta}- \frac{\log(\hat\xi_{q-t} -\ell+1)}{\log \beta} 
    \implies \beta \leq \frac{\hat\xi_{q+t} -\ell+1}{\hat\xi_{q-t} -\ell+1}.  
\end{align*}
This holds by assumption, and so a valid $k$ must exist.
We now consider $I$. 

Using the definition of $k$, the fact that $V$ and $\{V_i\}_{i=1}^\infty$ are standard normal random variables and a union bound, it holds that
{\small
\begin{align*}
\Prr{k<k_1}&\leq \sum_{i=1}^{k_1}\Prr{V_{i}/n\sqrt{\rho_2}\geq \hat{q}-F_n(\beta^{i}+\ell-1)}\\
&\leq \sum_{i=1}^{k_1}\Prr{V_{i}/n\sqrt{\rho_2}\geq \hat{q}'-F_n(\beta^{k_1}+\ell-1)}\\
&= \sum_{i=1}^{k_1}\Prr{V_{i}/n\sqrt{\rho_2}\geq q + V/n\sqrt{\rho_1} - F_n((\hat\xi_{q-t}-\ell+1)+\ell-1)}\\
&\leq k_1\Prr{V_{k_1}-V\geq tn\sqrt{(\rho_1\wedge \rho_2)}}\\
 &\leq \frac{\log(\hat \xi_{ q-t}+1-\ell)}{2\log\beta}\exp\left(-n(t\vee t^2)(\sqrt{\rho_1\vee \rho_1^2}\wedge\sqrt{\rho_2\vee \rho_2^2})/4 \right).
\end{align*}}
Here we use an exponential tail bound for the standard Gaussian, as we will be applying this result when $t<1$. 
Now, we consider $II$. 
Observe that 
\begin{align*}
II=\Prr{k> k_2}&= \prod_{i=1}^{k_2}\Prr{V_{i}/n\sqrt{\rho_2}+F_n(\beta^{i}+\ell-1)\leq \hat q}\\
&\leq\Prr{V_{i}/n\sqrt{\rho_2}+F_n(\beta^{k_2}+\ell-1)\leq \hat q}\\
&= \Prr{\frac{V_{i}}{n\sqrt{\rho_2}}+F_n((\hat\xi_{q+t}-\ell+1)+\ell-1)\leq q+\frac{V}{n\sqrt{\rho_1}}}\\
&\leq\Prr{tn\sqrt{(\rho_2\wedge\rho_2)} \leq V-V_{i}}\\
&\leq \exp\left(-n(t\vee t^2)(\sqrt{\rho_1\vee \rho_1^2}\wedge\sqrt{\rho_2\vee \rho_2^2})/4 \right).
\end{align*}
Then, for the second inequality, the same argument as in Lemma~\ref{lem::pq-bound} applies and the proof is complete. 
\end{proof}
\section{Proof of Corollary~\ref{corr:ssa}}
\begin{proof}[Proof of Corollary~\ref{corr:ssa}]
First, the triangle inequality implies $$|\tilde T(\bfX_n')'-T(F)|< |\mathrm{Bias}_{n,\eta}(T)|+|\tilde T(\bfX_n')'-\E(T(\bfX_n'^{(1)}))|.$$ Next, given that the conditions are satisfied by assumption, we apply Theorem~\ref{thm::main-result}. 
\end{proof}
\section{Error in Lugosi and Mendelson (2021)}\label{app::lm-error}
Here, we demonstrate that there is an error in the proof of Theorem 1 of \citet{Lugosi2021}. 
Namely, one step uses the inequalities 
$$\E\phi_{\xi_{2\zeta,\tilde X},\xi_{1-\zeta/2\tilde X}}(X)\leq \E(\tilde X\ind{X\geq \xi_{1-\zeta/2,X}}),$$
and
$$\E\phi_{\xi_{2\zeta,\tilde X},\xi_{1-\zeta/2,X}}(X)\geq -\E(\tilde X\ind{X\leq \xi_{2\zeta,X}}),$$
which do not hold, see page 400 and 401 of \citet{Lugosi2021}. 
We show they do not hold by counterexample. 
For the remainder of this section, we use the notation of \citet{Lugosi2021}. 

\subsection{The lower bound}
\noindent Using the notation in the paper, on page 401 at the very top, it states
\begin{align*}
    \mathbb{E}\phi_{\mu+Q_{2\epsilon}(\overline X),\mu+Q_{1-\epsilon/2}(\overline X)}( X)\geq \mu-\mathbb{E}[\overline X\mathbbm{1}_{\overline X\leq Q_{2\epsilon}(\overline X) }].
\end{align*}
This inequality does not hold. 
Observe that, using the fact that $\mathbb{E}(\overline X)=0$,
{\small
\begin{align*}
    \mathbb{E}\phi_{\mu+Q_{2\epsilon}(\overline X),\mu+Q_{1-\epsilon/2}(\overline X)}( X)-\mu&=Q_{2\epsilon}(\overline X) \mathbb{E}(\mathbbm{1}_{\overline X\leq Q_{2\epsilon}(\overline X)})+Q_{1-\epsilon/2}(\overline X)\mathbb{E}(\mathbbm{1}_{\overline X\geq Q_{1-\epsilon/2}(\overline X)})\\
    &\hspace{3em}+\mathbb{E}(\overline X(1-\mathbbm{1}_{\overline X\geq Q_{1-\epsilon/2}(\overline X)}-\mathbbm{1}_{\overline X\leq Q_{2\epsilon}(\overline X)}))\\
    &=Q_{2\epsilon}(\overline X) \mathbb{E}(\mathbbm{1}_{\overline X\leq Q_{2\epsilon}(\overline X)})+Q_{1-\epsilon/2}(\overline X)\mathbb{E}(\mathbbm{1}_{\overline X\geq Q_{1-\epsilon/2}(\overline X)})\\
    &\hspace{3em}-\mathbb{E}(\overline X\mathbbm{1}_{\overline X\geq Q_{1-\epsilon/2}(\overline X)})-\mathbb{E}(\overline X\mathbbm{1}_{\overline X\leq Q_{2\epsilon}(\overline X)})\\
    &= -\mathbb{E}(\overline X\mathbbm{1}_{\overline X\leq Q_{2\epsilon}(\overline X)})+ Q_{2\epsilon}(\overline X)\mathbb{E}(\mathbbm{1}_{\overline X\leq Q_{2\epsilon}(\overline X)})\\
    &\hspace{3em}+\mathbb{E}((Q_{1-\epsilon/2}(\overline X)-\overline X)\mathbbm{1}_{\overline X\geq Q_{1-\epsilon/2}(\overline X)}).
\end{align*}}
It then suffices to find some $F$ such that 
\begin{align*}
    Q_{2\epsilon}(\overline X)\mathbb{E}(\mathbbm{1}_{\overline X\leq Q_{2\epsilon}(\overline X)})+\mathbb{E}((Q_{1-\epsilon/2}(\overline X)-\overline X)\mathbbm{1}_{\overline X\geq Q_{1-\epsilon/2}(\overline X)})< 0. 
\end{align*}
The term $\mathbb{E}((Q_{1-\epsilon/2}(\overline X)-\overline X)\mathbbm{1}_{\overline X\geq Q_{1-\epsilon/2}(\overline X)})\leq 0$ for all $F$. 
Then, for example, standard normal distribution has $Q_{2\epsilon}(\overline X)\mathbb{E}(\mathbbm{1}_{\overline X\leq Q_{2\epsilon}(\overline X)})<0$, so the bound cannot hold in general. 

\subsection{The upper bound}
We now give a counterexample to 
\begin{align}  \label{eqn::bound}\mathbb{E}\phi_{\mu+Q_{2\epsilon}(\overline X),\mu+Q_{1-\epsilon/2}(\overline X)}( X)\leq \mu+\mathbb{E}[\overline X\mathbbm{1}_{\overline X> Q_{1-\epsilon/2}(\overline X) }],
\end{align}
for all $\epsilon<1/2$, $E(\overline X)=0$ and $F$ such that $Var(\overline X)<\infty$. 
Consider a random variable $\overline X\sim F$, where $F$ is a discrete distribution with four points, $x_1<x_2<0<x_3<x_4$. 
Suppose that $\Prr{X=x_i}=p_i$ for $i\in[4]$ and suppose that $\overline X$ has mean 0:  $\mathbb{E}(\overline X)=\sum_{i=1}^4x_ip_i=0$. 
If $x_i$ are finite, then the variance of $\overline X$ is finite. 
Therefore, $F$ is in the class of considered distributions. (Continuous analogues of this example can be derived, by placing concentrated density in a neighbourhood of $x_i$ for each $x_i$.) 
First, for $p_1<2\epsilon<p_2+p_1$ and $p_1+p_2<1-\epsilon/2\iff \epsilon<2(p_4+p_3)$, it holds that
\[\mathbb{E}\phi_{\mu+Q_{2\epsilon}(\overline X),\mu+Q_{1-\epsilon/2}(\overline X)} (\overline X)=(p_1+p_2)x_2+(p_3+p_4)x_3,\]
and, also,
$$\mathbb{E}[\overline X\mathbbm{1}_{\overline X> Q_{1-\epsilon/2}(\overline X) }]=x_4p_4.$$
To prove that \eqref{eqn::bound} does not hold for $F$, we must show that $((p_1+p_2)x_2+(p_3+p_4)x_3)>x_4p_4$ for appropriate $p_i$ and $x_i$. 
So we just need a solution to the set of constraints given by:
\begin{align*}
&0<\epsilon<1/2,\ \epsilon<2(p_4+p_3),\ p_1<2\epsilon<p_2+p_1,\ x_1<x_2<0<x_3<x_4,
\\
&p'1_4=1,\ x'p=0,\ 0<p_i<1,\ ((p_1+p_2)x_2+(p_3+p_4)x_3)>x_4p_4.
\end{align*}
Taking $p_1=0.01$,$p_2=0.6$,$p_3=0.38$,$p_4=0.01$, $\epsilon=0.05$, $x_1=-57.1$,$x_2=-1$, $x_3=3$,$x_4=3.1$ gives a solution. 
You can get a similar counterexample taking $F$ to be a continuous distribution which is a mixture of uniform densities on the neighborhoods $x_i\pm\delta$ for some small $\delta>0$, with weights $p_1,\ldots,p_4$. (You must also adjust $\epsilon$ by $\delta$). 

\section{Skewed distributions and trimmed means}
Here, we show that the non-private trimmed mean is inconsistent when $m/n\to p>0$ as $n\to \infty$ and the distribution is exponential. It follows by Slutsky's theorem that the private trimmed mean is also inconsistent. (The proof can easily be extended to any skewed distribution, but it is simpler to present a concrete case.)  
\begin{lemma}\label{lem::tm-inc}
Suppose that $\eta=0$ and $F$ is the standard exponential distribution. Then, when $m/n\to p>0$ as $n\to \infty$, the non-private trimmed mean is inconsistent. 
\end{lemma}
We have that, by the law of large numbers, the (non-private) trimmed mean converges in probability to $\E(X|X\in [-\log(1-p),-\log(p)])$ as $n\to\infty$, where $X\sim Exp(1)$. 
Observe that for $0<a<b<\infty$, it holds that
\begin{align*}
  \E(X|X\in [a,b])&=\int_{a}^{b} xe^{-x}dx\\
  &=-xe^{-x}|_{a}^{b}+ \int_{a}^{b} e^{-x}dx\\
  &=-be^{-b}+ae^{-a}-e^{-b}+e^{-a} \\
  &= (a+1)e^{-a} -(b+1)e^{-b}.
\end{align*}
It follows that $$\E(X|X\in [-\log(1-p),-\log(p)])=(1-\log(1-p))(1-p) -(1\log p )p\neq 1. $$
Therefore, for any non-vanishing clipping parameter $m=\floor{pn}$, the trimmed mean is inconsistent. It follows by Slutsky's theorem that the private trimmed mean is also inconsistent.

\section{Complete simulation results}
This section contains tables which summarize all of the simulation results.

\subsection{Subsample and aggregate full table}
This section contains the full table of subsample and aggregate results using other methods. It shows the empirical mean squared error between the non-private mixed model parameter estimates and each of the mixed model parameter estimates, produced by pairing subsample-and-aggregate algorithms with different private mean estimation algorithms, for different group sizes $k$ at $\rho=10$.
\begin{table*}[t]
\centering
\caption{Empirical mean squared error between the non-private mixed model parameter estimates and each of the mixed model parameter estimates, produced by pairing subsample-and-aggregate algorithms with different private mean estimation algorithms, for different group sizes $k$ at $\rho=1$. For large $k$, we are computing the mean of only a few observations ($<50$), and so it is more desirable to use the Huber estimator. For $k\geq 50$, the gains from $\tilde\mu_p'$ using are clear. }
\vspace{1em}
\resizebox{\linewidth}{!}{\begin{tabular}{c|ccccc}
\toprule
$k$ & $\tilde\mu_p'$ & Clipped mean & PRIME \cite{2021Liub} & Huber \citep{Yu2024} & Coinpress \citep{Biswas2020}  \\
\midrule
 10 &         \textbf{0.493} &         7.484 & 2289.532 &  3.075 &     81.140 \\
 20 &         \textbf{0.268} &        13.358 & 2302.300 &  3.200 &    813.758 \\
 30 &         \textbf{0.226} &        21.084 & 2307.890 &  3.216 &   3362.229 \\
 40 &         \textbf{0.235} &        28.901 & 2302.797 &  3.224 &   8906.546 \\
 50 &        \textbf{0.385} &        35.220 & 2299.656 &  3.229 &  18996.504 \\
 60 &         5.822 &        44.042 & 2308.754 &  \textbf{3.229} &  37248.857 \\
 70 &        23.988 &        50.617 & 2273.696 &  \textbf{3.231} &  54751.698 \\
 80 &        45.934 &        54.369 & 2296.854 &  \textbf{3.232} & 103898.736 \\
 90 &        72.352 &        64.545 & 2309.577 &  \textbf{3.233} & 163547.002 \\
100 &        88.130 &        70.247 & 2303.637 &  \textbf{3.232} & 213275.071 \\
\bottomrule
\end{tabular}}
\label{tab:diff-k}
\end{table*}
\subsection{Results comparing subsample and aggregate in PDP}\label{app::PDP}
The following table contains the mean absolute error of the coefficients in the mixed modelling example in Section~\ref{sec::data-analysis}, generated via the subsample-and-aggregate algorithm with $\tilde\mu_p$ and the widened winsorized mean of \citet{smith2011privacy} as the aggregators. 
We took $\epsilon=1$ as the privacy budget. 
\begin{table}[h!]
\begin{center}
\small
\begin{tabular}{lrr}
\toprule
$k$ &  $\tilde\mu_p$ &  Widened Winsorized Mean \citep{smith2011privacy} \\
\midrule
10  &          0.493 & 2152.134 \\
20  &          0.268 & 3376.355 \\
30  &          0.226 & 4062.483 \\
40  &          0.235 & 4214.284 \\
50  &          0.385 & 4835.852 \\
60  &          5.822 & 5244.282 \\
70  &         23.988 & 4875.672 \\
80  &         45.934 & 5606.984 \\
90  &         72.352 & 6624.289 \\
100 &         88.130 & 6475.178 \\
\bottomrule
\end{tabular}
\end{center}
\end{table}

\subsection{Results comparing estimators}\label{app::xtra-sims}
Below are the empirical mean squared errors of each private estimator, for each of the simulation scenarios.
\begin{figure}
    \centering
    \includegraphics[width=\textwidth]{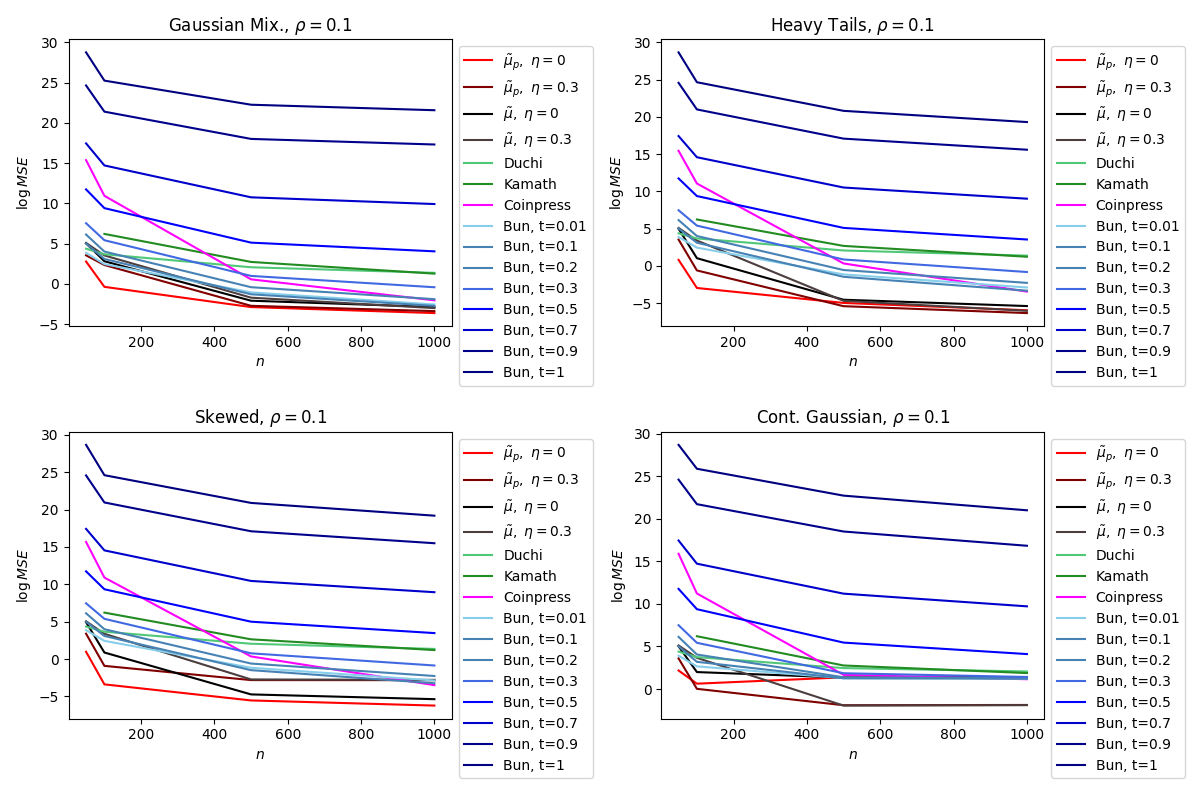}
    \caption{The empirical $\log MSE$ of all estimators, compared across a variety of different population distributions (Mixture of Gaussians, Skewed, Heavy-Tailed, and a Contaminated Gaussian) at $\rho=0.1$.}
    \label{fig:rhop1}
\end{figure}
\begin{figure}
    \centering
    \includegraphics[width=\textwidth]{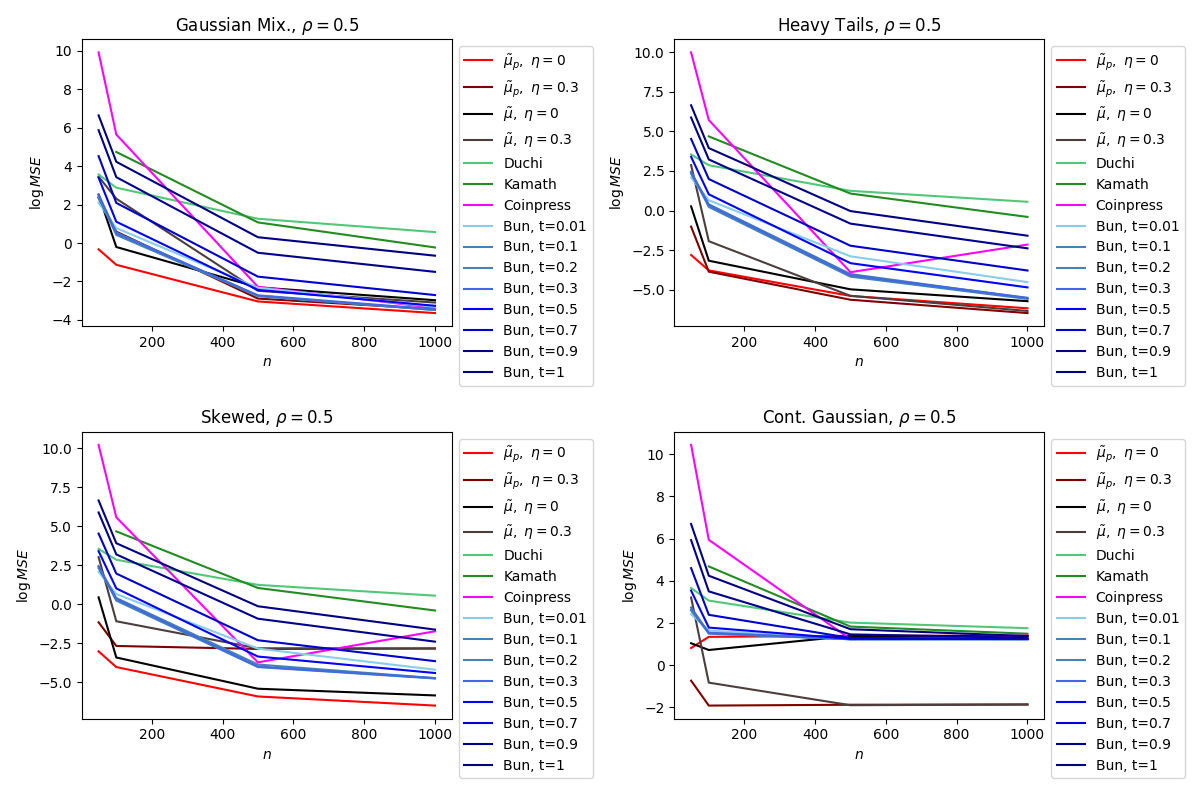}
    \caption{The empirical $\log MSE$ of all estimators, compared across a variety of different population distributions (Mixture of Gaussians, Skewed, Heavy-Tailed, and a Contaminated Gaussian) at $\rho=0.5$.}
    \label{fig:rhop5}
\end{figure}
\begin{figure}
    \centering
    \includegraphics[width=\textwidth]{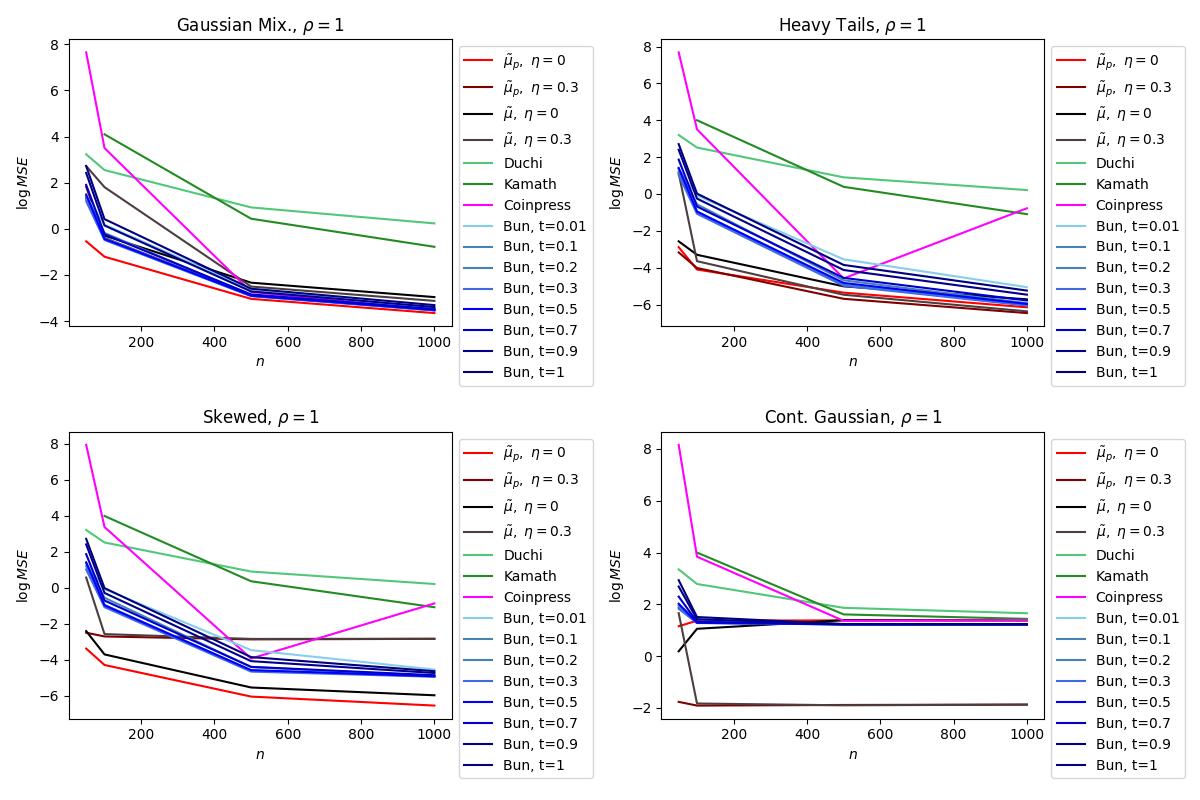}
    \caption{The empirical $\log MSE$ of all estimators, compared across a variety of different population distributions (Mixture of Gaussians, Skewed, Heavy-Tailed, and a Contaminated Gaussian) at $\rho=1$.}
    \label{fig:rho1}
\end{figure}
\begin{figure}
    \centering
    \includegraphics[width=\textwidth]{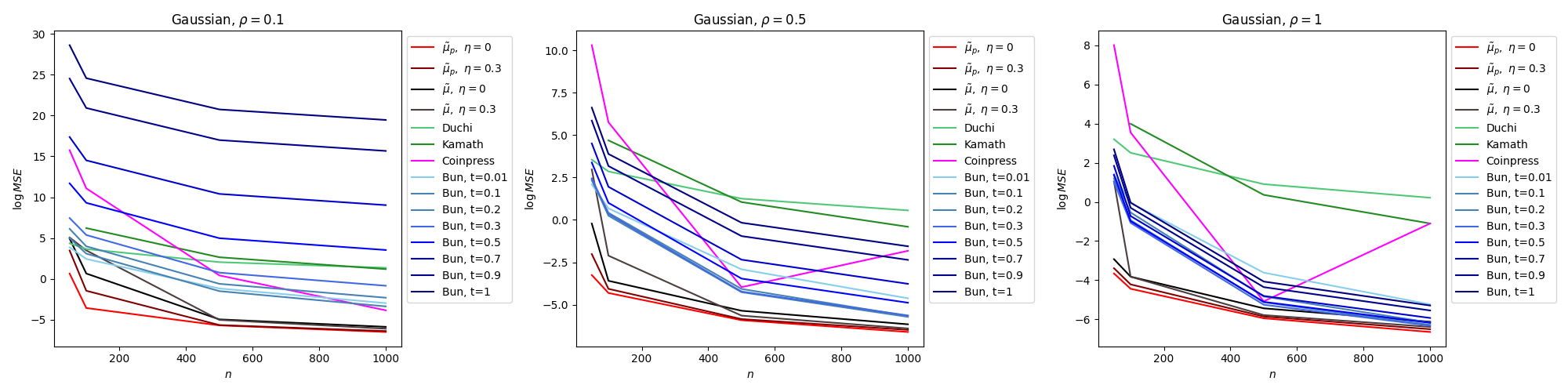}
    \caption{The empirical $\log MSE$ of all estimators, compared across a variety of different privacy budgets for the Gaussian distribution.}
    \label{fig:Gauss}
\end{figure}

\subsection{Results comparing the parameters $\rho$ and $C$}\label{app::crho}
Below are the empirical mean squared errors of $\tilde\mu_p'$, for different values of $\eta$ and $C$.

\begin{table}[p]
\centering
\scalebox{0.5}{\begin{tabular}{lrrrrrrr}
\hline
\multicolumn{6}{c}{$n=50$}\\
\hline
Distribution & $\rho$& $C$& $\eta$=0.0 & $\eta$=0.05 & $\eta$=0.15& $\eta$=0.3 \\
\hline
Gaussian & 0.1 & 5 & 6.9979 & 6.2420 & 12.8326 & 23.8501 \\
Gaussian & 0.1 & 10 & 6.9979 & 6.2420 & 12.8326 & 23.8501 \\
Gaussian & 0.1 & 100 & 6.9979 & 6.2420 & 12.8326 & 23.8501 \\
Gaussian Mix. & 0.1 & 5 & 16.1027 & 15.8987 & 25.2483 & 31.5837 \\
Gaussian Mix. & 0.1 & 10 & 16.1027 & 15.8987 & 25.2483 & 31.5837 \\
Gaussian Mix. & 0.1 & 100 & 16.1027 & 15.8987 & 25.2483 & 31.5837 \\
Skewed & 0.1 & 5 & 7.9827 & 5.6974 & 12.1654 & 25.6706 \\
Skewed & 0.1 & 10 & 7.9827 & 5.6974 & 12.1654 & 25.6706 \\
Skewed & 0.1 & 100 & 7.9827 & 5.6974 & 12.1654 & 25.6706 \\
Heavy Tails & 0.1 & 5 & 5.0993 & 5.9785 & 13.3371 & 23.3229 \\
Heavy Tails & 0.1 & 10 & 5.0993 & 5.9785 & 13.3371 & 23.3229 \\
Heavy Tails & 0.1 & 100 & 5.0993 & 5.9785 & 13.3371 & 23.3229 \\
Cont. Gaussian & 0.1 & 5 & 10.6869 & 8.0391 & 21.0777 & 22.6884 \\
Cont. Gaussian & 0.1 & 10 & 10.6869 & 8.0391 & 21.0777 & 22.6884 \\
Cont. Gaussian & 0.1 & 100 & 10.6869 & 8.0391 & 21.0777 & 22.6884 \\
Gaussian & 0.5 & 5 & 0.0335 & 0.0344 & 0.0362 & 0.3177 \\
Gaussian & 0.5 & 10 & 0.0335 & 0.0344 & 0.0362 & 0.3177 \\
Gaussian & 0.5 & 100 & 0.0335 & 0.0344 & 0.0362 & 0.3177 \\
Gaussian Mix. & 0.5 & 5 & 0.5763 & 0.8150 & 2.3005 & 9.5372 \\
Gaussian Mix. & 0.5 & 10 & 0.5763 & 0.8150 & 2.3005 & 9.5372 \\
Gaussian Mix. & 0.5 & 100 & 0.5763 & 0.8150 & 2.3005 & 9.5372 \\
Skewed & 0.5 & 5 & 0.0454 & 0.0454 & 0.0656 & 0.1375 \\
Skewed & 0.5 & 10 & 0.0454 & 0.0454 & 0.0656 & 0.1375 \\
Skewed & 0.5 & 100 & 0.0454 & 0.0454 & 0.0656 & 0.1375 \\
Heavy Tails & 0.5 & 5 & 0.0574 & 0.0617 & 0.0524 & 0.3773 \\
Heavy Tails & 0.5 & 10 & 0.0574 & 0.0617 & 0.0524 & 0.3773 \\
Heavy Tails & 0.5 & 100 & 0.0574 & 0.0617 & 0.0524 & 0.3773 \\
Cont. Gaussian & 0.5 & 5 & 2.0844 & 1.6988 & 0.5707 & 0.4226 \\
Cont. Gaussian & 0.5 & 10 & 2.0844 & 1.6988 & 0.5707 & 0.4226 \\
Cont. Gaussian & 0.5 & 100 & 2.0844 & 1.6988 & 0.5707 & 0.4226 \\
Gaussian & 1 & 5 & 0.0279 & 0.0275 & 0.0298 & 0.0322 \\
Gaussian & 1 & 10 & 0.0279 & 0.0275 & 0.0298 & 0.0322 \\
Gaussian & 1 & 100 & 0.0279 & 0.0275 & 0.0298 & 0.0322 \\
Gaussian Mix. & 1 & 5 & 0.5726 & 0.5364 & 0.7384 & 6.3019 \\
Gaussian Mix. & 1 & 10 & 0.5726 & 0.5364 & 0.7384 & 6.3019 \\
Gaussian Mix. & 1 & 100 & 0.5726 & 0.5364 & 0.7384 & 6.3019 \\
Skewed & 1 & 5 & 0.0320 & 0.0335 & 0.0532 & 0.0860 \\
Skewed & 1 & 10 & 0.0320 & 0.0335 & 0.0532 & 0.0860 \\
Skewed & 1 & 100 & 0.0320 & 0.0335 & 0.0532 & 0.0860 \\
Heavy Tails & 1 & 5 & 0.0532 & 0.0501 & 0.0461 & 0.0438 \\
Heavy Tails & 1 & 10 & 0.0532 & 0.0501 & 0.0461 & 0.0438 \\
Heavy Tails & 1 & 100 & 0.0532 & 0.0501 & 0.0461 & 0.0438 \\
Cont. Gaussian & 1 & 5 & 3.1484 & 2.5389 & 0.6734 & 0.1771 \\
Cont. Gaussian & 1 & 10 & 3.1484 & 2.5389 & 0.6734 & 0.1771 \\
Cont. Gaussian & 1 & 100 & 3.1484 & 2.5389 & 0.6734 & 0.1771 \\
\end{tabular}}
\end{table}

\begin{table}[p]
\centering
\scalebox{0.5}{\begin{tabular}{lrrrrrrr}
\hline
\multicolumn{6}{c}{$n=100$}\\
\hline
Distribution & $\rho$& $C$& $\eta$=0.0 & $\eta$=0.05 & $\eta$=0.15& $\eta$=0.3 \\
\hline
Gaussian & 0.1 & 5 & 0.0304 & 0.0269 & 0.0269 & 0.5741 \\
Gaussian & 0.1 & 10 & 0.0304 & 0.0269 & 0.0269 & 0.5741 \\
Gaussian & 0.1 & 100 & 0.0304 & 0.0269 & 0.0269 & 0.5741 \\
Gaussian Mix. & 0.1 & 5 & 0.5193 & 0.6651 & 2.9036 & 9.5608 \\
Gaussian Mix. & 0.1 & 10 & 0.5193 & 0.6651 & 2.9036 & 9.5608 \\
Gaussian Mix. & 0.1 & 100 & 0.5193 & 0.6651 & 2.9036 & 9.5608 \\
Skewed & 0.1 & 5 & 0.0417 & 0.0407 & 0.0720 & 0.1096 \\
Skewed & 0.1 & 10 & 0.0417 & 0.0407 & 0.0720 & 0.1096 \\
Skewed & 0.1 & 100 & 0.0417 & 0.0407 & 0.0720 & 0.1096 \\
Heavy Tails & 0.1 & 5 & 0.0362 & 0.0315 & 0.0342 & 0.8930 \\
Heavy Tails & 0.1 & 10 & 0.0362 & 0.0315 & 0.0342 & 0.8930 \\
Heavy Tails & 0.1 & 100 & 0.0362 & 0.0315 & 0.0342 & 0.8930 \\
Cont. Gaussian & 0.1 & 5 & 1.9036 & 1.4556 & 0.5181 & 1.0141 \\
Cont. Gaussian & 0.1 & 10 & 1.9036 & 1.4556 & 0.5181 & 1.0141 \\
Cont. Gaussian & 0.1 & 100 & 1.9036 & 1.4556 & 0.5181 & 1.0141 \\
Gaussian & 0.5 & 5 & 0.0140 & 0.0131 & 0.0151 & 0.0171 \\
Gaussian & 0.5 & 10 & 0.0140 & 0.0131 & 0.0151 & 0.0171 \\
Gaussian & 0.5 & 100 & 0.0140 & 0.0131 & 0.0151 & 0.0171 \\
Gaussian Mix. & 0.5 & 5 & 0.2881 & 0.2903 & 0.2890 & 1.3460 \\
Gaussian Mix. & 0.5 & 10 & 0.2881 & 0.2903 & 0.2890 & 1.3460 \\
Gaussian Mix. & 0.5 & 100 & 0.2881 & 0.2903 & 0.2890 & 1.3460 \\
Skewed & 0.5 & 5 & 0.0167 & 0.0201 & 0.0367 & 0.0733 \\
Skewed & 0.5 & 10 & 0.0167 & 0.0201 & 0.0367 & 0.0733 \\
Skewed & 0.5 & 100 & 0.0167 & 0.0201 & 0.0367 & 0.0733 \\
Heavy Tails & 0.5 & 5 & 0.0246 & 0.0214 & 0.0196 & 0.0177 \\
Heavy Tails & 0.5 & 10 & 0.0246 & 0.0214 & 0.0196 & 0.0177 \\
Heavy Tails & 0.5 & 100 & 0.0246 & 0.0214 & 0.0196 & 0.0177 \\
Cont. Gaussian & 0.5 & 5 & 3.8366 & 3.6247 & 1.0255 & 0.1711 \\
Cont. Gaussian & 0.5 & 10 & 3.8366 & 3.6247 & 1.0255 & 0.1711 \\
Cont. Gaussian & 0.5 & 100 & 3.8366 & 3.6247 & 1.0255 & 0.1711 \\
Gaussian & 1 & 5 & 0.0131 & 0.0129 & 0.0134 & 0.0151 \\
Gaussian & 1 & 10 & 0.0131 & 0.0129 & 0.0134 & 0.0151 \\
Gaussian & 1 & 100 & 0.0131 & 0.0129 & 0.0134 & 0.0151 \\
Gaussian Mix. & 1 & 5 & 0.2739 & 0.2883 & 0.2802 & 0.5840 \\
Gaussian Mix. & 1 & 10 & 0.2739 & 0.2883 & 0.2802 & 0.5840 \\
Gaussian Mix. & 1 & 100 & 0.2739 & 0.2883 & 0.2802 & 0.5840 \\
Skewed & 1 & 5 & 0.0146 & 0.0165 & 0.0328 & 0.0707 \\
Skewed & 1 & 10 & 0.0146 & 0.0165 & 0.0328 & 0.0707 \\
Skewed & 1 & 100 & 0.0146 & 0.0165 & 0.0328 & 0.0707 \\
Heavy Tails & 1 & 5 & 0.0234 & 0.0222 & 0.0180 & 0.0161 \\
Heavy Tails & 1 & 10 & 0.0234 & 0.0222 & 0.0180 & 0.0161 \\
Heavy Tails & 1 & 100 & 0.0234 & 0.0222 & 0.0180 & 0.0161 \\
Cont. Gaussian & 1 & 5 & 3.9678 & 3.9259 & 1.4978 & 0.1658 \\
Cont. Gaussian & 1 & 10 & 3.9678 & 3.9259 & 1.4978 & 0.1658 \\
Cont. Gaussian & 1 & 100 & 3.9678 & 3.9259 & 1.4978 & 0.1658 \\
\end{tabular}}
\end{table}

\begin{table}[p]
\centering
\scalebox{0.5}{\begin{tabular}{lrrrrrrr}
\hline
\multicolumn{6}{c}{$n=500$}\\
\hline
Distribution & $\rho$& $C$& $\eta$=0.0 & $\eta$=0.05 & $\eta$=0.15& $\eta$=0.3 \\
\hline
Gaussian & 0.1 & 5 & 0.0028 & 0.0023 & 0.0025 & 0.0030 \\
Gaussian & 0.1 & 10 & 0.0026 & 0.0023 & 0.0025 & 0.0030 \\
Gaussian & 0.1 & 100 & 0.0026 & 0.0023 & 0.0025 & 0.0030 \\
Gaussian Mix. & 0.1 & 5 & 0.0567 & 0.0628 & 0.0561 & 0.0667 \\
Gaussian Mix. & 0.1 & 10 & 0.0604 & 0.0628 & 0.0561 & 0.0667 \\
Gaussian Mix. & 0.1 & 100 & 0.0606 & 0.0628 & 0.0561 & 0.0667 \\
Skewed & 0.1 & 5 & 0.0037 & 0.0064 & 0.0217 & 0.0585 \\
Skewed & 0.1 & 10 & 0.0043 & 0.0064 & 0.0217 & 0.0585 \\
Skewed & 0.1 & 100 & 0.0048 & 0.0064 & 0.0217 & 0.0585 \\
Heavy Tails & 0.1 & 5 & 0.0059 & 0.0045 & 0.0043 & 0.0042 \\
Heavy Tails & 0.1 & 10 & 0.0057 & 0.0045 & 0.0043 & 0.0042 \\
Heavy Tails & 0.1 & 100 & 0.0053 & 0.0045 & 0.0043 & 0.0042 \\
Cont. Gaussian & 0.1 & 5 & 3.9958 & 3.9681 & 2.6566 & 0.1544 \\
Cont. Gaussian & 0.1 & 10 & 3.9773 & 3.9681 & 2.6566 & 0.1544 \\
Cont. Gaussian & 0.1 & 100 & 3.9572 & 3.9681 & 2.6566 & 0.1544 \\
Gaussian & 0.5 & 5 & 0.0021 & 0.0019 & 0.0021 & 0.0024 \\
Gaussian & 0.5 & 10 & 0.0021 & 0.0019 & 0.0021 & 0.0024 \\
Gaussian & 0.5 & 100 & 0.0021 & 0.0019 & 0.0021 & 0.0024 \\
Gaussian Mix. & 0.5 & 5 & 0.0510 & 0.0515 & 0.0544 & 0.0614 \\
Gaussian Mix. & 0.5 & 10 & 0.0520 & 0.0515 & 0.0544 & 0.0614 \\
Gaussian Mix. & 0.5 & 100 & 0.0535 & 0.0515 & 0.0544 & 0.0614 \\
Skewed & 0.5 & 5 & 0.0027 & 0.0046 & 0.0196 & 0.0575 \\
Skewed & 0.5 & 10 & 0.0028 & 0.0046 & 0.0196 & 0.0575 \\
Skewed & 0.5 & 100 & 0.0030 & 0.0046 & 0.0196 & 0.0575 \\
Heavy Tails & 0.5 & 5 & 0.0047 & 0.0037 & 0.0033 & 0.0033 \\
Heavy Tails & 0.5 & 10 & 0.0043 & 0.0037 & 0.0033 & 0.0033 \\
Heavy Tails & 0.5 & 100 & 0.0049 & 0.0037 & 0.0033 & 0.0033 \\
Cont. Gaussian & 0.5 & 5 & 3.9963 & 3.9483 & 3.6557 & 0.1539 \\
Cont. Gaussian & 0.5 & 10 & 3.9841 & 3.9483 & 3.6557 & 0.1539 \\
Cont. Gaussian & 0.5 & 100 & 3.9816 & 3.9483 & 3.6557 & 0.1539 \\
Gaussian & 1 & 5 & 0.0020 & 0.0019 & 0.0020 & 0.0023 \\
Gaussian & 1 & 10 & 0.0020 & 0.0019 & 0.0020 & 0.0023 \\
Gaussian & 1 & 100 & 0.0020 & 0.0019 & 0.0020 & 0.0023 \\
Gaussian Mix. & 1 & 5 & 0.0503 & 0.0526 & 0.0540 & 0.0604 \\
Gaussian Mix. & 1 & 10 & 0.0498 & 0.0526 & 0.0540 & 0.0604 \\
Gaussian Mix. & 1 & 100 & 0.0511 & 0.0526 & 0.0540 & 0.0604 \\
Skewed & 1 & 5 & 0.0026 & 0.0044 & 0.0193 & 0.0576 \\
Skewed & 1 & 10 & 0.0027 & 0.0044 & 0.0193 & 0.0576 \\
Skewed & 1 & 100 & 0.0028 & 0.0044 & 0.0193 & 0.0576 \\
Heavy Tails & 1 & 5 & 0.0046 & 0.0038 & 0.0032 & 0.0031 \\
Heavy Tails & 1 & 10 & 0.0046 & 0.0038 & 0.0032 & 0.0031 \\
Heavy Tails & 1 & 100 & 0.0044 & 0.0038 & 0.0032 & 0.0031 \\
Cont. Gaussian & 1 & 5 & 3.9928 & 3.9633 & 3.6730 & 0.1528 \\
Cont. Gaussian & 1 & 10 & 3.9917 & 3.9633 & 3.6730 & 0.1528 \\
Cont. Gaussian & 1 & 100 & 3.9759 & 3.9633 & 3.6730 & 0.1528 \\
\end{tabular}}
\end{table}

\begin{table}[p]
\centering
\scalebox{0.5}{\begin{tabular}{lrrrrrrr}
\hline
\multicolumn{6}{c}{$n=1000$}\\
\hline
Distribution & $\rho$& $C$& $\eta$=0.0 & $\eta$=0.05 & $\eta$=0.15& $\eta$=0.3 \\
\hline
Gaussian & 0.1 & 5 & 0.0013 & 0.0012 & 0.0012 & 0.0015 \\
Gaussian & 0.1 & 10 & 0.0013 & 0.0012 & 0.0012 & 0.0015 \\
Gaussian & 0.1 & 100 & 0.0013 & 0.0012 & 0.0012 & 0.0015 \\
Gaussian Mix. & 0.1 & 5 & 0.0282 & 0.0284 & 0.0265 & 0.0306 \\
Gaussian Mix. & 0.1 & 10 & 0.0274 & 0.0284 & 0.0265 & 0.0306 \\
Gaussian Mix. & 0.1 & 100 & 0.0252 & 0.0284 & 0.0265 & 0.0306 \\
Skewed & 0.1 & 5 & 0.0015 & 0.0040 & 0.0199 & 0.0581 \\
Skewed & 0.1 & 10 & 0.0016 & 0.0040 & 0.0199 & 0.0581 \\
Skewed & 0.1 & 100 & 0.0022 & 0.0040 & 0.0199 & 0.0581 \\
Heavy Tails & 0.1 & 5 & 0.0038 & 0.0027 & 0.0025 & 0.0020 \\
Heavy Tails & 0.1 & 10 & 0.0029 & 0.0027 & 0.0025 & 0.0020 \\
Heavy Tails & 0.1 & 100 & 0.0032 & 0.0027 & 0.0025 & 0.0020 \\
Cont. Gaussian & 0.1 & 5 & 3.9839 & 3.9805 & 3.6247 & 0.1511 \\
Cont. Gaussian & 0.1 & 10 & 4.0048 & 3.9805 & 3.6247 & 0.1511 \\
Cont. Gaussian & 0.1 & 100 & 3.9785 & 3.9805 & 3.6247 & 0.1511 \\
Gaussian & 0.5 & 5 & 0.0011 & 0.0011 & 0.0011 & 0.0013 \\
Gaussian & 0.5 & 10 & 0.0011 & 0.0011 & 0.0011 & 0.0013 \\
Gaussian & 0.5 & 100 & 0.0010 & 0.0011 & 0.0011 & 0.0013 \\
Gaussian Mix. & 0.5 & 5 & 0.0263 & 0.0261 & 0.0261 & 0.0299 \\
Gaussian Mix. & 0.5 & 10 & 0.0257 & 0.0261 & 0.0261 & 0.0299 \\
Gaussian Mix. & 0.5 & 100 & 0.0250 & 0.0261 & 0.0261 & 0.0299 \\
Skewed & 0.5 & 5 & 0.0012 & 0.0035 & 0.0195 & 0.0588 \\
Skewed & 0.5 & 10 & 0.0013 & 0.0035 & 0.0195 & 0.0588 \\
Skewed & 0.5 & 100 & 0.0018 & 0.0035 & 0.0195 & 0.0588 \\
Heavy Tails & 0.5 & 5 & 0.0030 & 0.0025 & 0.0022 & 0.0018 \\
Heavy Tails & 0.5 & 10 & 0.0028 & 0.0025 & 0.0022 & 0.0018 \\
Heavy Tails & 0.5 & 100 & 0.0027 & 0.0025 & 0.0022 & 0.0018 \\
Cont. Gaussian & 0.5 & 5 & 3.9892 & 3.9629 & 3.6655 & 0.1509 \\
Cont. Gaussian & 0.5 & 10 & 3.9886 & 3.9629 & 3.6655 & 0.1509 \\
Cont. Gaussian & 0.5 & 100 & 3.9818 & 3.9629 & 3.6655 & 0.1509 \\
Gaussian & 1 & 5 & 0.0010 & 0.0011 & 0.0011 & 0.0013 \\
Gaussian & 1 & 10 & 0.0011 & 0.0011 & 0.0011 & 0.0013 \\
Gaussian & 1 & 100 & 0.0011 & 0.0011 & 0.0011 & 0.0013 \\
Gaussian Mix. & 1 & 5 & 0.0257 & 0.0262 & 0.0258 & 0.0297 \\
Gaussian Mix. & 1 & 10 & 0.0258 & 0.0262 & 0.0258 & 0.0297 \\
Gaussian Mix. & 1 & 100 & 0.0249 & 0.0262 & 0.0258 & 0.0297 \\
Skewed & 1 & 5 & 0.0012 & 0.0033 & 0.0195 & 0.0591 \\
Skewed & 1 & 10 & 0.0012 & 0.0033 & 0.0195 & 0.0591 \\
Skewed & 1 & 100 & 0.0017 & 0.0033 & 0.0195 & 0.0591 \\
Heavy Tails & 1 & 5 & 0.0028 & 0.0024 & 0.0021 & 0.0017 \\
Heavy Tails & 1 & 10 & 0.0027 & 0.0024 & 0.0021 & 0.0017 \\
Heavy Tails & 1 & 100 & 0.0024 & 0.0024 & 0.0021 & 0.0017 \\
Cont. Gaussian & 1 & 5 & 3.9926 & 3.9593 & 3.6732 & 0.1513 \\
Cont. Gaussian & 1 & 10 & 3.9845 & 3.9593 & 3.6732 & 0.1513 \\
Cont. Gaussian & 1 & 100 & 3.9818 & 3.9593 & 3.6732 & 0.1513 \\
\end{tabular}}
\end{table}

\subsection{Results concerning the privacy budget distribution}\label{app::pbudg}
Below are the empirical mean squared errors of $\tilde\mu_p'$, for different distributions of $\rho$, across the quantile estimation step and the mean estimation step.

\begin{table}[p]
\centering
\scalebox{0.7}{\begin{tabular}{lrrrr}
\hline
\multicolumn{5}{c}{$n=50$}\\
\hline
Distribution &  $\rho$ & Quantile favored & Equal & Mean favored \\
\hline
Gaussian & 0.1 & 1.0724 & 3.2019 & 26.9747 \\
Gaussian Mix. & 0.1 & 12.3804 & 18.1274 & 31.0007 \\
Skewed & 0.1 & 0.5160 & 2.0156 & 28.2555 \\
Heavy Tails & 0.1 & 0.7095 & 2.1220 & 25.6965 \\
Cont. Gaussian & 0.1 & 4.4354 & 4.7744 & 30.4394 \\
Gaussian & 0.5 & 0.0767 & 0.0527 & 0.0571 \\
Gaussian Mix. & 0.5 & 1.0573 & 0.8350 & 3.5427 \\
Skewed & 0.5 & 0.0441 & 0.0622 & 0.0654 \\
Heavy Tails & 0.5 & 0.0781 & 0.0692 & 0.0669 \\
Cont. Gaussian & 0.5 & 2.6438 & 1.9936 & 1.1525 \\
Gaussian & 1 & 0.0358 & 0.0332 & 0.0334 \\
Gaussian Mix. & 1 & 0.7170 & 0.6965 & 0.7507 \\
Skewed & 1 & 0.0308 & 0.0293 & 0.0452 \\
Heavy Tails & 1 & 0.0681 & 0.0496 & 0.0522 \\
Cont. Gaussian & 1 & 3.4910 & 3.0430 & 1.9391 \\
\hline
\multicolumn{5}{c}{$n=100$}\\
\hline
Distribution &  $\rho$ & Quantile favored & Equal & Mean favored \\
\hline
Gaussian & 0.1 & 0.0594 & 0.0516 & 0.0679 \\
Gaussian Mix. & 0.1 & 0.7165 & 0.7130 & 3.7338 \\
Skewed & 0.1 & 0.0552 & 0.0436 & 0.0659 \\
Heavy Tails & 0.1 & 0.0945 & 0.0679 & 0.0632 \\
Cont. Gaussian & 0.1 & 2.1998 & 1.7083 & 1.1326 \\
Gaussian & 0.5 & 0.0195 & 0.0161 & 0.0129 \\
Gaussian Mix. & 0.5 & 0.3353 & 0.3217 & 0.3026 \\
Skewed & 0.5 & 0.0177 & 0.0145 & 0.0191 \\
Heavy Tails & 0.5 & 0.0348 & 0.0302 & 0.0368 \\
Cont. Gaussian & 0.5 & 3.9762 & 3.8108 & 3.0390 \\
Gaussian & 1 & 0.0123 & 0.0121 & 0.0130 \\
Gaussian Mix. & 1 & 0.2806 & 0.2693 & 0.2857 \\
Skewed & 1 & 0.0119 & 0.0128 & 0.0133 \\
Heavy Tails & 1 & 0.0284 & 0.0250 & 0.0265 \\
Cont. Gaussian & 1 & 4.0326 & 3.9716 & 3.8126 \\
\end{tabular}}
\end{table}
\begin{table}[p]
\centering
\scalebox{0.7}{\begin{tabular}{lrrrr}
\hline
\multicolumn{5}{c}{$n=500$}\\
\hline
Distribution &  $\rho$ & Quantile favored & Equal & Mean favored \\
\hline
Gaussian & 0.1 & 0.0054 & 0.0036 & 0.0028 \\
Gaussian Mix. & 0.1 & 0.0675 & 0.0639 & 0.0631 \\
Skewed & 0.1 & 0.0044 & 0.0042 & 0.0046 \\
Heavy Tails & 0.1 & 0.0061 & 0.0058 & 0.0067 \\
Cont. Gaussian & 0.1 & 4.0150 & 3.9968 & 3.9681 \\
Gaussian & 0.5 & 0.0028 & 0.0025 & 0.0025 \\
Gaussian Mix. & 0.5 & 0.0572 & 0.0554 & 0.0556 \\
Skewed & 0.5 & 0.0030 & 0.0030 & 0.0031 \\
Heavy Tails & 0.5 & 0.0056 & 0.0061 & 0.0053 \\
Cont. Gaussian & 0.5 & 4.0198 & 4.0195 & 3.9965 \\
Gaussian & 1 & 0.0024 & 0.0023 & 0.0023 \\
Gaussian Mix. & 1 & 0.0547 & 0.0549 & 0.0548 \\
Skewed & 1 & 0.0028 & 0.0027 & 0.0029 \\
Heavy Tails & 1 & 0.0060 & 0.0051 & 0.0059 \\
Cont. Gaussian & 1 & 4.0029 & 4.0041 & 4.0079 \\
\hline
\multicolumn{5}{c}{$n=1000$}\\
\hline
Distribution &  $\rho$ & Quantile favored & Equal & Mean favored \\
\hline
Gaussian & 0.1 & 0.0018 & 0.0015 & 0.0014 \\
Gaussian Mix. & 0.1 & 0.0282 & 0.0274 & 0.0281 \\
Skewed & 0.1 & 0.0020 & 0.0019 & 0.0022 \\
Heavy Tails & 0.1 & 0.0025 & 0.0023 & 0.0023 \\
Cont. Gaussian & 0.1 & 4.0193 & 4.0102 & 3.9968 \\
Gaussian & 0.5 & 0.0012 & 0.0011 & 0.0011 \\
Gaussian Mix. & 0.5 & 0.0266 & 0.0258 & 0.0258 \\
Skewed & 0.5 & 0.0015 & 0.0015 & 0.0016 \\
Heavy Tails & 0.5 & 0.0026 & 0.0022 & 0.0023 \\
Cont. Gaussian & 0.5 & 3.9962 & 4.0047 & 3.9987 \\
Gaussian & 1 & 0.0011 & 0.0011 & 0.0011 \\
Gaussian Mix. & 1 & 0.0256 & 0.0260 & 0.0258 \\
Skewed & 1 & 0.0015 & 0.0015 & 0.0015 \\
Heavy Tails & 1 & 0.0027 & 0.0022 & 0.0021 \\
Cont. Gaussian & 1 & 3.9928 & 3.9931 & 4.0017 \\
\end{tabular}}
\end{table}

\end{document}